%% file: main-v2.tex
\documentclass[11pt]{article}

\usepackage[margin=1in]{geometry}

\usepackage{amsmath,amssymb,amsthm}
\usepackage{graphicx}
\usepackage{booktabs}
\usepackage{longtable,array}
\usepackage{float}
\usepackage{xr-hyper}
\usepackage{xurl}
\usepackage[hidelinks]{hyperref}
\usepackage{cite}
\newtheorem{lemma}{Lemma}
\newtheorem{theorem}{Theorem}
\newtheorem{proposition}{Proposition}
\newtheorem{corollary}{Corollary}

\begin{document}

\title{A Theory of Information Architecture for Networked Decisions:
Freshness, Locality, and Coordination}
\author{
Scott Moeller\\
\small Independent Researcher, Portsmouth, RI, USA\\
\small \texttt{scott.h.moeller@gmail.com}
\and
Bhaskar Krishnamachari\\
\small University of Southern California\\
\small \texttt{bkrishna@usc.edu}
}

\date{}

\hypersetup{
  pdftitle={A Theory of Information Architecture for Networked Decisions: Freshness, Locality, and Coordination},
  pdfauthor={Scott Moeller and Bhaskar Krishnamachari},
  pdfsubject={Decision-relevant information scope and freshness in networked static teams},
  pdfkeywords={information architecture, networked decisions, static teams, age of information, distributed coordination}
}

\maketitle

\begin{abstract}
A broader view of a network can help agents coordinate, but collecting that view takes time. When is a fresh local observation more useful than an older view of the whole network, and when should an agent look only partway across it? We study agents that repeatedly minimize a shared quadratic cost in an environment that evolves independently of their actions. An information architecture specifies what each agent observes and how old those observations are. Each agent uses a synchronized neighborhood snapshot, so extending its scope also ages the nearby readings used in the decision. We assume that the optimal action with complete information is affine in the state, and that the conditional mean of the current state equals the old state scaled by a common exponential decay factor. Using the classical Marschak--Radner decomposition, we derive the delay threshold beyond which fresh local information outperforms a complete but older global view. This comparison allows non-Gaussian states and nonlinear decision rules. We then determine how far an agent should look by balancing the benefit of additional observations against their greater age. In a Gaussian spatial model on a line, we derive the optimal observation radius: looking beyond the local node pays only when information travels more than twice the spatial correlation length in one coherence time. The preferred information architecture also depends on the decision being made. On a fixed 64-node ring, increasing only the reach of the decision weights from 0.7 to 6 hops shifts the optimal observation radius from one to five hops.
\end{abstract}

\section{Introduction}
\label{sec:introduction}

Many networked systems must repeatedly make coupled decisions using information
that is distributed across space and changes over time. A wireless transmitter
may know its current local channel state but not the current interference
conditions elsewhere in the network; an edge server may observe its own
workload immediately while a system-wide view is available only after telemetry
has been collected and disseminated; and a distributed sensing system may have
fresh measurements from nearby sensors while observations from distant sensors
arrive with greater latency. In each case, broad information can improve
coordination, but obtaining it takes time. The resulting design choice is not
simply between centralized and distributed control, but between
\emph{information architectures} that differ simultaneously in spatial scope
and temporal freshness.

\paragraph*{A two-sensor example.}
Two sensors report independent signals $X_t$ and $Y_t$, each with mean zero
and variance $\sigma^2$. We want to estimate their current average
$Z_t=(X_t+Y_t)/2$. Which is more useful: one sensor's reading now, or both
sensors' readings from $\tau$ seconds ago?

With only $X_t$, our best estimate is $X_t/2$. The error is the missing term
$Y_t/2$, so the mean squared error is $\sigma^2/4$. With both old readings,
we know the old average $Z_{t-\tau}$ exactly. Suppose the two signal processes are independent and stationary, and each
satisfies the scalar conditional-mean prediction rule
$\mathbb E[X_t\mid X_{t-\tau}]=e^{-\tau/T}X_{t-\tau}$ (and likewise
for $Y$), where $T$ is the coherence time. Our best estimate today is then $e^{-\tau/T}Z_{t-\tau}$.
Its mean squared error is
$\operatorname{Var}(Z_t)(1-e^{-2\tau/T})$. Since
$\operatorname{Var}(Z_t)=\sigma^2/2$, the fresh reading wins when
\[
    \frac{\sigma^2}{4}
    <\frac{\sigma^2}{2}(1-e^{-2\tau/T})
    \quad\Longleftrightarrow\quad
    \tau>\frac{T\log 2}{2}\approx 0.35T.
\]
Thus, one current reading predicts the two-sensor average better than both
readings once the pair is more than about $0.35T$ old. The older snapshot
contains the entire spatial state, but its additional scope no longer
compensates for its age. A similar crossover arises when agents' actions
affect one another's costs: fresh local information can outperform a
delayed global view even when the agents must coordinate.

For the main results, each agent acts on a synchronized snapshot: widening
scope can replace a fresh local reading by an older one. Architectures that
retain fresh local data while adding remote data have different information
and fall outside this comparison.
Figure~\ref{fig:snapshot-architectures} shows why a wider synchronized
snapshot is older: readings from farther away need more time to reach
the deciding agent.

\begin{figure}[H]
\centering
\includegraphics[width=0.90\linewidth]{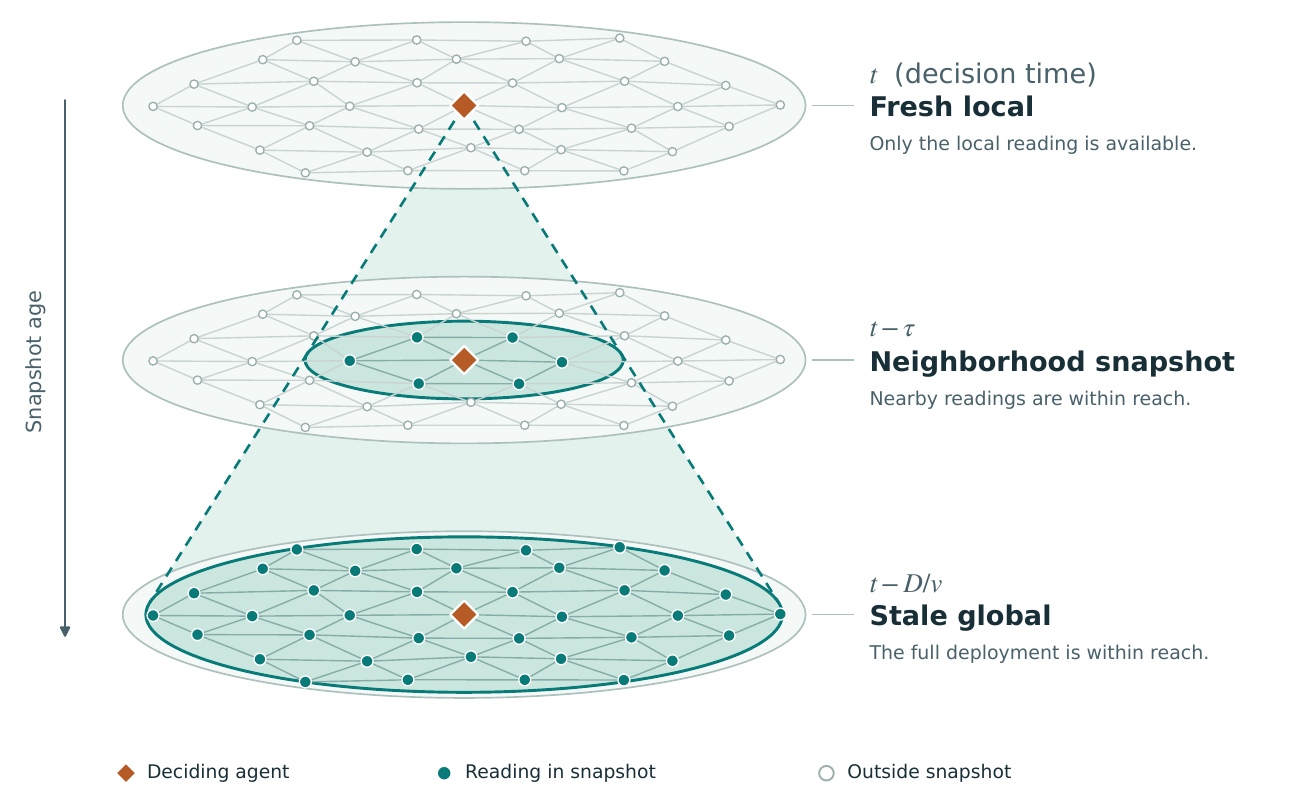}
\caption{Information propagation in a planar network. The tip of the cone marks
the deciding agent (diamond) at decision time $t$. Each horizontal slice
shows a separate synchronized snapshot. For propagation speed $v$, a
snapshot of age $\tau$ can include readings from agents within distance
$v\tau$. At the tip, only the current local reading is available. At age
$D/v$, readings from the entire deployment can arrive, where $D$ is the
distance from the deciding agent to the farthest agent. Every reading in
that global snapshot, including the deciding agent's own, is $D/v$ old.}
\label{fig:snapshot-architectures}
\end{figure}

The example raises two design questions. How much delay can a global view
tolerate before fresh local information becomes preferable? And when
neither endpoint is best, how wide a neighborhood should an agent observe?
Both questions depend on which parts of the state matter to the decision:
information about a remote node is useful to the extent that it helps
predict the action the agent should take.

We consider agents whose actions incur a shared quadratic cost in an
environment that evolves independently of those actions. An information
architecture specifies which state observations each agent can use and how
old they are. We compare a fresh local reading, a delayed global view, and
synchronized neighborhood snapshots by how well they support the current
full-information decision. Current actions affect neither future state nor
any agent's later information; this is a repeated static-team problem.

Within this snapshot family, a wider view is worthwhile only while its
additional decision information compensates for its greater age. We
quantify this tradeoff through a local--global crossover and an optimal
neighborhood radius. Stronger decision coupling can make a global view
worth waiting longer for, while changing which remote states matter to
the desired action can change the preferred radius even on the same network.

The contributions are:
\begin{enumerate}
\item \textbf{When does fresh local information beat a stale global view?}
We characterize the crossover by an explicit threshold on delay relative
to coherence time. The threshold depends on the fraction of decision loss
eliminated by fresh local observations. In a shared-resource model,
stronger coupling increases the delay worth tolerating for a global view.

\item \textbf{How far should an agent look?}
We characterize the optimal synchronized neighborhood by balancing the
gain from broader scope against the loss from older observations. For a
Gaussian field on a line, we derive an explicit radius in terms of
spatial correlation, decision reach, and information propagation. The
result identifies when the optimum is local and when a positive radius
is worthwhile.

\item \textbf{Why can different decisions require different information architectures?}
We show how the desired action determines the value of remote observations.
On a fixed 64-node ring, changing only the decision weights moves the
preferred radius from one to five hops. We also evaluate the radius
criterion on path, ring, grid, and random geometric graphs.
\end{enumerate}

We build on the classical team theory of Radner and of Marschak and
Radner \cite{radner1962team,marschak1972economic}, which asks how agents
with different information can best pursue a shared objective. Their
delayed-information decomposition provides the analytical basis for our
comparison. We use it to evaluate synchronized neighborhoods through
their ability to predict the desired action, derive the spatial omission
curve and optimal radius for a Gaussian field on a line, and compare
decision-dependent scopes on finite graphs.

Related questions arise in stale-information analysis
\cite{mitzenmacher2000old}, age of information \cite{kaul2012realtime},
distributed optimization \cite{nedic2009distributed}, organizational
economics \cite{alonso2008coordination}, and networked control with
finite-speed communication
\cite{bamieh2002spatial,voulgaris2003optimal,bamieh2005convex,fardad2011finite}.
Scope-dependent communication delay can also make sparse feedback
outperform all-to-all control, as shown by Ballotta, Jovanovi\'c, and
Schenato \cite{ballotta2023latency}. Here we compare information
architectures for repeated decisions in an exogenous environment,
valuing each architecture by the error it induces in the desired action.
Section~\ref{sec:related_work} develops these connections.

Section~\ref{sec:model} defines observations, regret, and the temporal
prediction assumption. Section~\ref{sec:spatiotemporal_predictability}
proves the composition law and derives the endpoint comparison and coupling
example. Section~\ref{sec:optimal_coordination_scale} chooses the radius.
Section~\ref{sec:numerical-evaluation} gives four numerical illustrations,
including noncanonical graphs; Sections~\ref{sec:related_work} and
\ref{sec:discussion_conclusions} discuss related work and implications.
Technical appendices retain the broader formulation, proofs, and
reproducibility details.

\section{Model and observations}
\label{sec:model}

We compare what agents can achieve with different observations of the same
environment. The environment evolves independently of their actions, and
actions do not affect later observations. Each epoch is therefore a static
team decision. Actions are unconstrained Euclidean vectors.

\subsection{The desired action and its cost}
\label{subsec:decision_model}

There are $n$ agents on a graph $G=(V,E)$. Its hop distance
$d_G(i,j)$ defines the neighborhood
\begin{equation}
\mathcal N_r(i)=\{j\in V:d_G(i,j)\le r\}.
\label{eq:r_neighborhood}
\end{equation}
The graph describes observation scope and acquisition delay. The stacked
state $\theta_t=(\theta_{1,t},\ldots,\theta_{n,t})$ has finite second
moment; local states and local actions may each be vectors. Write
$\Sigma_S=\operatorname{Cov}(\theta_t)$ for the stationary state covariance.
We center known means throughout the main argument. Subtracting the mean
state and mean desired action is a change of coordinates, not a restriction
to physically zero-mean environments.

With complete current information, the desired action is $x_t^\star=K\theta_t$.
We measure excess cost by the coupled tracking loss
\begin{equation}
x_t^\star=K\theta_t,\qquad
\ell(x,\theta_t)=\frac12(x-x_t^\star)^\top Q(x-x_t^\star),\qquad Q\succ0.
\label{eq:tracking_loss}
\end{equation}
Here $K$ is a fixed matrix describing how state affects the desired action,
and $Q$ is deterministic, symmetric, and positive definite. Its off-diagonal
blocks allow agents' action errors to interact in the shared cost. Neither
matrix must follow the communication graph. Completing the square in a
quadratic objective with affine state dependence gives exactly this loss
after centering; Appendix~\ref{app:technical-projection} gives that equivalence
and the broader fixed-epoch formulation with non-affine desired actions.

Figure~\ref{fig:field-and-cost} illustrates how $K$ maps a changing
spatial field to a moving target action and how $Q$ determines the
cost of missing that target.

\begin{figure}[H]
\centering
\includegraphics[width=0.70\linewidth]{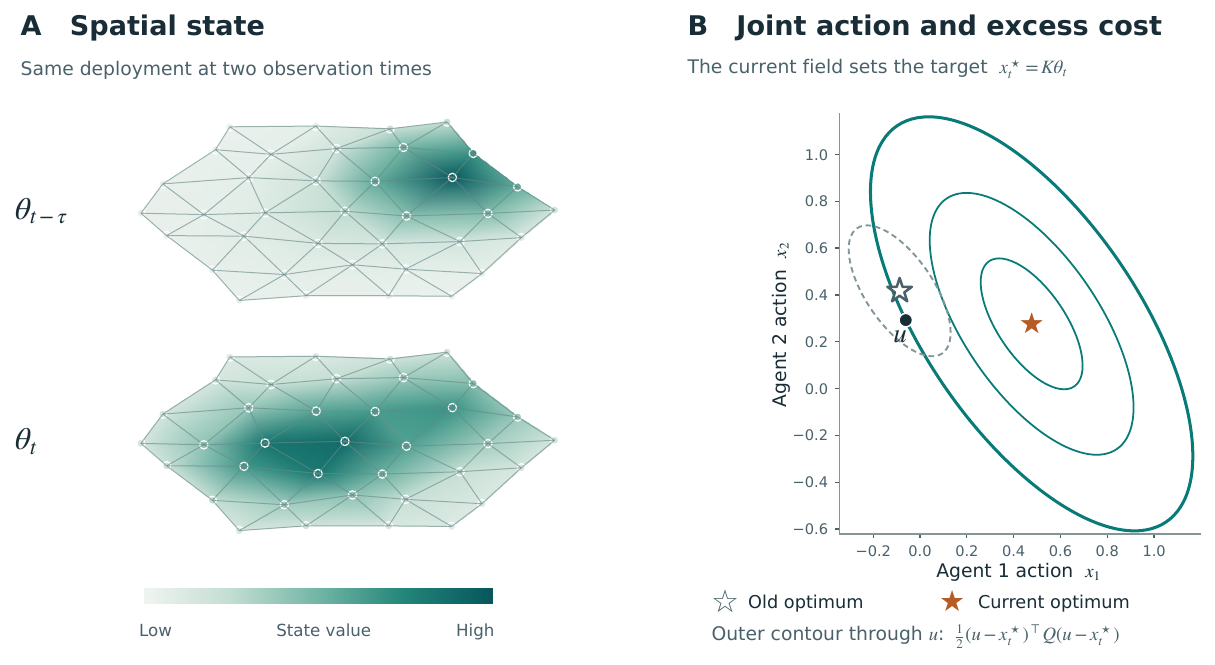}
\caption{From spatial state to decision cost. Left: two snapshots of the same
field, shown on a common color scale. Right: as the field changes, the
full-information action $x_t^\star=K\theta_t$ moves from the hollow star
to the filled star. The dashed ellipse and the innermost solid ellipse
represent the same excess cost at the two times. The dot marks the joint
action $u$ chosen from the older snapshot. The solid contour through $u$
gives its excess cost for the current field.}
\label{fig:field-and-cost}
\end{figure}

\paragraph{Decision relevance.}
\label{subsec:decision_relevant_smoothness}
A state change $\xi$ changes the desired action by $K\xi$. If $K\xi=0$,
that change creates no decision error. Other directions matter according
to $(K\xi)^\top Q(K\xi)$. In particular,
\begin{equation}
\operatorname{Cov}(x_t^\star)=K\Sigma_S K^\top.
\label{eq:decision_covariance}
\end{equation}
The value of observations depends jointly on $K$, $Q$, the state covariance,
and which states are observed. Figure~\ref{fig:decision-relevance} illustrates
this dependence with the graph and state covariance held fixed.

\subsection{What each agent observes}
\label{subsec:architectures}

At radius $r$ and age $\tau$, agent $i$ observes the vector
\begin{equation}
Z_{i,t}^{(r,\tau)}=(\theta_{j,t-\tau}:j\in\mathcal N_r(i)).
\label{eq:radius_arch}
\end{equation}
We denote this architecture by $\mathcal A_{r,\tau}$. The two parameters
are independent here; Section~\ref{sec:optimal_coordination_scale} later
sets $\tau=\tau(r)$. The three cases are:
\begin{itemize}
\item \textbf{Fresh local:} each agent observes its current state $\theta_{i,t}$,
giving $\mathcal A_{0,0}$ and regret $R_{\rm loc}$.
\item \textbf{Stale global:} every agent observes the same complete old state
$\theta_{t-\tau}$, giving regret $R_{\rm glob}(\tau)$.
\item \textbf{Synchronized neighborhood:} each agent observes
$Z_{i,t}^{(r,\tau)}$, so every reading used has the same timestamp.
\end{itemize}
An admissible rule is $u_i=g_i(Z_{i,t}^{(r,\tau)})$ with finite second
moment. We allow all such rules, including nonlinear ones.

Widening a synchronized snapshot may replace nearby fresh readings with
older readings so that all observations have a common timestamp. Retaining
fresh local readings while adding old remote readings gives a different
architecture. Adding usable information without discarding existing
information cannot increase optimal cost, because agents can ignore it.
The same observation applies to a delayed global summary added to local data;
such mixed-age observations are outside the snapshot family studied here.

\subsection{Regret and the fresh omission share}
\label{subsec:architecture_regret}

Regret is the smallest expected excess cost achievable when each agent uses
only its specified observations:
\begin{equation}
R(\mathcal A)=\inf_{\{g_i\}}\frac12
\mathbb E[(u-x_t^\star)^\top Q(u-x_t^\star)].
\label{eq:architecture_regret}
\end{equation}
This optimizes over admissible rules, not a particular distributed algorithm.
It uses the fixed-information static-team benchmark
\cite{radner1962team,krainak1982staticII}; the formal existence, uniqueness,
and projection result is in Appendix~\ref{app:technical-projection}.
With a common observation, the optimum is the vector conditional mean of
$x_t^\star$. With different observations and off-diagonal $Q$, separate
agentwise conditional means generally do not minimize team cost.

The best no-information action is zero in centered coordinates. Its loss is
\begin{equation}
R_\infty=\frac12\mathbb E[(x_t^\star)^\top Qx_t^\star]
=\frac12\operatorname{tr}(K^\top QK\Sigma_S).
\label{eq:Rinf_general}
\end{equation}
For $R_\infty>0$, define
\begin{equation}
\eta_S(r)=\frac{R(\mathcal A_{r,0})}{R_\infty},\qquad 0\le\eta_S(r)\le1.
\label{eq:eta_S_definition}
\end{equation}
This is the fraction of no-information decision loss remaining with a fresh
radius-$r$ view. Larger fresh neighborhoods cannot increase it, since their
extra observations can be ignored. Radius zero gives
$\eta_S(0)=R_{\rm loc}/R_\infty$; at the diameter of a finite connected
graph the fresh view is complete and omission is zero.

\subsection{A common temporal prediction rule}
\label{subsec:OU_model}

Assume a stationary state process with finite second moments and, for every
delay analyzed, the transition
\begin{equation}
\theta_t=a\theta_{t-\tau}+\varepsilon_{t,\tau},\qquad
a=e^{-\tau/T},\quad T>0,\qquad
\mathbb E[\varepsilon_{t,\tau}\mid\theta_{t-\tau}]=0.
\label{eq:state_innovation_decomposition}
\end{equation}
The best mean prediction of current state from the complete old state is
the old state multiplied by a common decay factor. The remaining disturbance
has mean zero even after the \emph{whole} old state is known. This assumption
requires neither Gaussian states nor independence of the disturbance from the old state.
State coordinates may be spatially dependent. We write $\rho=\tau/T$ for
delay measured in coherence times.

Stationarity supplies the same one-time distribution at the two epochs.
Equality of covariance matrices alone would not identify optimal errors
over arbitrary nonlinear rules. Likewise, a covariance decay relation alone
does not imply the conditional-mean property in
Eq.~\eqref{eq:state_innovation_decomposition}.

Imagine two edge servers whose average workload changes slowly while their
workload difference changes quickly. An old report can remain useful for
total demand after losing value for assigning the next job. Our common-rate
assumption assigns both patterns the same decay time $T$; distinct rates
are outside the exact factorization below. A stationary Gaussian
Ornstein--Uhlenbeck process is one sufficient special case
\cite{ornee2021ou}, described in Appendix~\ref{app:temporal-special-cases}.
A stationary AR(1) process with coefficient $0<c<1$ and independent,
mean-zero finite-variance innovations gives the analogous prediction
coefficient $c^k$ at integer lag $k$.

\section{How scope and freshness combine}
\label{sec:spatiotemporal_predictability}
\label{sec:fresh_local_global}

How much of a fresh neighborhood's value survives when its readings are
delayed? Under the common-rate prediction assumption, the surviving
fraction is the square of the temporal prediction factor. Marschak and
Radner derived the corresponding delayed-information decomposition for
quadratic teams with scalar autoregression
\cite[Ch.~7, Sec.~4, Eqs.~(4.21)--(4.22)]{marschak1972economic}.\footnote{Their
$w^d$ is the prediction factor $a$ below, their $\widehat V(0)$ corresponds
to $R_\infty$, and their $L(0)$ is the regret with the same information at
zero delay.}
We state it here for synchronized neighborhoods under the conditional-mean
assumptions of Section~\ref{subsec:OU_model}.

\begin{theorem}[Spatio-temporal composition]
\label{thm:spatiotemporal_composition}
Under the centered affine tracking model \eqref{eq:tracking_loss}, stationary
conditional-mean transition \eqref{eq:state_innovation_decomposition}, and
synchronized observations \eqref{eq:radius_arch}, every radius $r$ and finite
delay $\tau\ge0$ satisfy
\begin{equation}
\boxed{R(\mathcal A_{r,\tau})=(1-a^2)R_\infty+a^2R(\mathcal A_{r,0}),
\qquad a=e^{-\tau/T}.}
\label{eq:regret_composition_unnormalized}
\end{equation}
The optimization allows all admissible finite-second-moment rules and any
deterministic $Q\succ0$. For $R_\infty>0$, equivalently,
\begin{equation}
\boxed{1-\frac{R(\mathcal A_{r,\tau})}{R_\infty}
=e^{-2\tau/T}[1-\eta_S(r)].}
\label{eq:regret_composition_normalized}
\end{equation}
\end{theorem}
The fraction of no-information loss eliminated equals the fraction captured
by a fresh neighborhood times the fraction of predictive value surviving
the delay. Spatial covariance is unrestricted under the stated temporal
assumption.

\begin{proof}
Applying $K$ to the transition gives
\begin{equation}
x_t^\star=a x_{t-\tau}^\star+K\varepsilon_{t,\tau}.
\label{eq:xstar_innovation_decomposition}
\end{equation}
Every joint admissible rule $u$ is a function of the complete old state,
although its components each use only their own neighborhood. On expanding
the squared tracking error, the cross term vanishes:
\begin{align*}
&\mathbb E[(a x_{t-\tau}^\star-u)^\top QK\varepsilon_{t,\tau}]\\
&\quad=\mathbb E[(a x_{t-\tau}^\star-u)^\top QK
\mathbb E[\varepsilon_{t,\tau}\mid\theta_{t-\tau}]]=0.
\end{align*}
Thus the error is the sum of a predictable-part error and a disturbance error:
\begin{align}
\mathbb E[(x_t^\star-u)^\top Q(x_t^\star-u)]
={}&\mathbb E[(a x_{t-\tau}^\star-u)^\top Q(a x_{t-\tau}^\star-u)]\nonumber\\
&+\mathbb E[(K\varepsilon_{t,\tau})^\top Q(K\varepsilon_{t,\tau})].
\label{eq:innovation_split}
\end{align}
Apply the same expansion with $u=0$. Stationarity gives
\begin{equation}
\mathbb E[(K\varepsilon_{t,\tau})^\top Q(K\varepsilon_{t,\tau})]
=(1-a^2)\mathbb E[(x_t^\star)^\top Qx_t^\star]=2(1-a^2)R_\infty.
\label{eq:innovation_energy}
\end{equation}
For finite delay $a>0$, write $u=a\widetilde u$. This rescaling is a
bijection of the admissible rule class: each component uses exactly the
same observations and still has finite second moment. Consequently,
\begin{align}
&\inf_u\mathbb E[(a x_{t-\tau}^\star-u)^\top Q(a x_{t-\tau}^\star-u)]\nonumber\\
&\quad=a^2\inf_{\widetilde u}\mathbb E[(x_{t-\tau}^\star-\widetilde u)^\top
Q(x_{t-\tau}^\star-\widetilde u)]
=2a^2R(\mathcal A_{r,0}).
\label{eq:policy_rescaling}
\end{align}
The last equality uses stationarity and the same observation pattern at
each epoch. Combining the two errors and dividing by two proves the law;
normalization is valid when $R_\infty>0$.
\end{proof}

The proof rescales a joint rule; it does not replace a coupled team optimum
by separate conditional means. The disturbance since the snapshot is
unpredictable from any of its neighborhoods. A wider neighborhood can
reduce only the remaining, temporally predictable part of the error.

\subsection{Complete old information versus current local information}
\label{subsec:two_architectures}

\begin{corollary}[Exact stale-global regret]
\label{thm:OU_stale_regret}
Under Theorem~\ref{thm:spatiotemporal_composition}, a complete snapshot of age
$\tau$ gives
\begin{equation}
R_{\rm glob}(\tau)=(1-e^{-2\tau/T})R_\infty.
\label{eq:OU_exact_regret}
\end{equation}
\end{corollary}
\begin{proof}
With a complete fresh snapshot the desired action is feasible, so
$R(\mathcal A_{r,0})=0$ in Eq.~\eqref{eq:regret_composition_unnormalized}.
\end{proof}
The optimal global action is $aK\theta_{t-\tau}$. Its regret is zero at
zero delay, grows initially as $2(\tau/T)R_\infty+o(\tau/T)$, and tends to
$R_\infty$ as delay grows. This last statement is a limit, not an invertible
policy rescaling at infinite delay.

If different parts of the state change at different rates, a complete
old snapshot can still have an explicit value. Suppose it reveals
independent stationary Ornstein--Uhlenbeck modes with coherence times $T_k$,
and let $w_k$ be the fraction of the no-information decision loss
$R_\infty>0$ attributable to mode $k$, with $\sum_k w_k=1$. Then
\[
1-\frac{R_{\rm glob}(\tau)}{R_\infty}
=\sum_k w_k e^{-2\tau/T_k}.
\]
In the two-server example in Section~\ref{subsec:OU_model}, an old snapshot
remains useful longer for estimating total workload if total workload
changes more slowly than the difference between the two loads. This formula
applies to a complete old snapshot. Extending it to partial neighborhood
observations requires additional assumptions.

\begin{theorem}[Fresh-local versus stale-global crossover]
\label{thm:local_global_crossover}
Under Theorem~\ref{thm:spatiotemporal_composition}, let $R_\infty>0$ and
$0\le\eta_S(0)<1$. Fresh-local information is strictly better exactly when
\begin{equation}
\boxed{\frac{\tau}{T}>\rho^\star
\triangleq\frac12\log\frac{1}{1-\eta_S(0)}.}
\label{eq:OU_crossover}
\end{equation}
Equality gives equal regret. If $\eta_S(0)=1$, no finite delay gives a
strict local advantage; we write $\rho^\star=+\infty$.
\end{theorem}
\begin{proof}
Divide the two regrets by $R_\infty$: the comparison is
$\eta_S(0)<1-e^{-2\tau/T}$. Rearranging gives the threshold and equality case.
\end{proof}
When $\eta_S(0)=0$, the local view already determines the desired action,
and every positive global delay makes local information preferable.
For an omission share strictly between zero and one, the two losses cross
once. The threshold concerns prediction of the desired action, rather than
raw state correlation alone.

\subsection{Role of Decision Coupling}
\label{subsec:coupling_example}

The fresh-local omission share $\eta_S(0)$ depends jointly on the
spatial statistics of the environment and on the coupling structure of the
decision problem. To isolate the role of decision coupling, consider a
symmetric shared-resource model with $n\ge2$ scalar actions,
\begin{equation}
    f(x,\theta)
    =
    \sum_{i=1}^{n}
    \left(
        \frac{q}{2}x_i^2-\theta_i x_i
    \right)
    +
    \frac{\gamma}{2}
    \left(
        \sum_{i=1}^{n}x_i
    \right)^2,
    \qquad
    q>0,\;\gamma\ge0.
    \label{eq:shared_resource_cost}
\end{equation}
Here $q$ determines the local curvature or stiffness of each decision, while
$\gamma$ determines the strength of the shared coordination penalty. The cost has
\begin{equation}
    Q
    =
    qI+\gamma\mathbf{1}\mathbf{1}^\top,
    \qquad
    K=Q^{-1}.
    \label{eq:shared_resource_Q}
\end{equation}

Here $\mathbf1$ is the all-ones vector.

Assume that the state components $\theta_i$ are independent and mean-zero,
with equal finite variance $\sigma^2>0$. This calculation concerns a single epoch.

The full-information action $x^\star=Q^{-1}\theta$ has components

\begin{equation}
    x_i^\star
    =
    \frac{\theta_i}{q}
    -
    \frac{\gamma}
    {q(q+\gamma n)}
    \sum_{j=1}^{n}\theta_j.
    \label{eq:shared_resource_xstar}
\end{equation}
Thus, even though the linear reward $\theta_i x_i$ is local, coupling through
the common resource makes agent $i$'s optimal action depend on the states of
all other agents.

\begin{proposition}[Local explainability under aggregate coupling]
\label{prop:aggregate_local}
For the model \eqref{eq:shared_resource_cost}, the optimal fresh-local policy
is
\begin{equation}
    u_i^{\mathrm{loc}}
    =
    \frac{\theta_i}{q+\gamma},
    \label{eq:optimal_local_shared}
\end{equation}
and the fresh-local omission share is
\begin{equation}
    \boxed{
    \eta_S(0)
    =
    \frac{
        \gamma^2(n-1)
    }{
        (q+\gamma)
        \left(
            q+\gamma(n-1)
        \right)
    }.
    }
    \label{eq:aggregate_eta}
\end{equation}
\end{proposition}

\begin{proof}
For any admissible $u_i=g_i(\theta_i)$, write
$u_i=\mu_i+\widetilde u_i$, where $\mu_i=\mathbb E[u_i]$ and
$\mathbb E[\widetilde u_i]=0$. Independence across agents and zero state
means give
\begin{align}
\mathbb E[f(u,\theta)]={}&
\sum_i\mathbb E\left[\frac{q+\gamma}{2}\widetilde u_i^2
-\theta_i\widetilde u_i\right]
+\frac q2\sum_i\mu_i^2+\frac\gamma2\left(\sum_i\mu_i\right)^2.
\label{eq:aggregate_centered_cost}
\end{align}
The mean terms are minimized at zero. Each remaining term is minimized
pointwise by $\widetilde u_i=\theta_i/(q+\gamma)$, which is feasible and
has zero mean. Thus the expected local cost is
$-n\sigma^2/[2(q+\gamma)]$, and the best constant cost is zero.
The eigenvalues of $Q$ are $q+\gamma n$ once and $q$ with multiplicity
$n-1$. The full-information expected cost is therefore
$-(\sigma^2/2)[(n-1)/q+1/(q+\gamma n)]$.
Subtracting this cost from the local and constant costs gives
Eqs.~\eqref{eq:aggregate_Rinf}--\eqref{eq:aggregate_Rloc} in
Appendix~\ref{app:technical-aggregate}; their ratio is the stated omission share.
\end{proof}

Under the common-rate temporal assumption, combining Proposition~\ref{prop:aggregate_local} with
Theorem~\ref{thm:local_global_crossover} gives an explicit threshold.

\begin{equation}
    \boxed{
    \rho^\star(\gamma,n)
    =
    \frac{1}{2}
    \log
    \left[
        \frac{
            (q+\gamma)
            \left(
                q+\gamma(n-1)
            \right)
        }{
            q(q+\gamma n)
        }
    \right].
    }
    \label{eq:aggregate_rho_star}
\end{equation}
Therefore
\begin{equation}
    R_{\mathrm{loc}}
    <
    R_{\mathrm{glob}}(\tau)
    \quad\Longleftrightarrow\quad
    \frac{\tau}{T}
    >
    \rho^\star(\gamma,n).
    \label{eq:aggregate_crossover}
\end{equation}

As $n\rightarrow\infty$,
\begin{equation}
    \boxed{
    \rho^\star(\gamma,n)
    \longrightarrow
    \frac{1}{2}
    \log
    \left(
        1+\frac{\gamma}{q}
    \right).
    }
    \label{eq:large_n_threshold}
\end{equation}

In the large-system limit, local information captures the share
$q/(q+\gamma)$ of decision variation; the remainder requires nonlocal state.

When $\gamma=0$, the decisions decouple and
$\eta_S(0)=0$: current local information is sufficient, so any
positive global delay makes the fresh-local architecture preferable. As $\gamma$
increases, each agent's optimal action depends more strongly on remote states,
$\eta_S(0)$ increases, and the crossover moves to larger
$\tau/T$. Thus stronger coordination coupling makes the system willing to
tolerate increasingly stale-global information before abandoning global
coordination in favor of fresh-local decisions.

\section{Choosing the observation radius}
\label{sec:optimal_coordination_scale}

We first ask when a wider synchronized snapshot is worth its extra delay,
then compute the optimal radius in a concrete spatial model.

\subsection{When is one more hop worthwhile?}
\label{subsec:radius_dependent_architectures}

Let $D$ be the maximum allowed radius and let
$\tau:[0,D]\to[0,\infty)$ be nondecreasing, with $\tau(0)=0$.
Define the fraction of no-information loss eliminated by the snapshot as
\begin{equation}
G(r)=\underbrace{1-\eta_S(r)}_{\text{share captured by a fresh snapshot}}
     \;\underbrace{e^{-2\tau(r)/T}}_{\text{fraction surviving the delay}}.
\label{eq:snapshot_gain}
\end{equation}
Writing $R(r)=R(\mathcal A_{r,\tau(r)})$,
Theorem~\ref{thm:spatiotemporal_composition} gives
$R(r)=R_\infty[1-G(r)]$, so the best radius maximizes $G(r)$.

Under the synchronized-snapshot convention of Section~\ref{subsec:architectures}, widening the neighborhood trades more spatial information for older readings.

On a finite graph, one more hop improves performance exactly when
\begin{equation}
G(r+1)>G(r).
\label{eq:discrete_radius_improvement}
\end{equation}
Compare all admissible integer radii on a general graph; stopping at the
first decline requires unimodality. Likewise, rounding a continuous optimum
to a neighboring integer is justified only for a unimodal interpolated
regret curve. Radius zero gives fresh-local information, while the diameter
of a finite connected graph gives a complete stale-global snapshot.

The loss separates into the two components plotted in
Fig.~\ref{fig:radius-regret}:
\begin{equation}
\frac{R(r)}{R_\infty}
=\underbrace{1-e^{-2\tau(r)/T}}_{\text{loss from delay}}
+\underbrace{e^{-2\tau(r)/T}\eta_S(r)}_{\text{remaining loss from limited scope}}.
\label{eq:R_decomposition}
\end{equation}
The delay term grows with radius, while the remaining scope term falls.

\begin{proposition}[Continuous marginal-balance characterization]
\label{thm:marginal_balance}
Under the assumptions of
Theorem~\ref{thm:spatiotemporal_composition}, additionally assume that
$\tau(r)$ and $\eta_S(r)$ are continuously differentiable on $[0,D]$, with
\[
    \tau'(r)\ge0,\qquad \eta_S'(r)\le0,\qquad 0\le\eta_S(r)<1.
\]
Every interior optimal radius $r^\star\in(0,D)$ satisfies
\begin{equation}
\left.\frac{d}{dr}\log[1-\eta_S(r)]\right|_{r=r^\star}
=\frac{2\tau'(r^\star)}{T}.
\label{eq:marginal_balance_condition}
\end{equation}
\end{proposition}

\begin{proof}
Since $G(r)>0$, maximizing $G$ is equivalent to maximizing $\log G$.
At an interior optimum, $(\log G)'=0$; differentiating
Eq.~\eqref{eq:snapshot_gain} gives the stated equality.
\end{proof}

The proportional gain from wider scope equals the proportional loss from
delay. If the difference between the left and right sides of
Eq.~\eqref{eq:marginal_balance_condition}, evaluated at radius $r$, is
strictly decreasing in $r$, its sign identifies the unique optimum:
radius zero if the difference is nonpositive at zero, radius $D$ if it is
nonnegative at $D$, and otherwise the unique interior crossing.
Holding the spatial omission curve fixed and preserving this single-crossing
condition, increasing $T$ moves the preferred radius weakly outward, while
increasing the latency scale in $\tau(r)=\alpha\tau_0(r)$ moves it weakly inward.
The optimum can stay unchanged at a boundary. These conclusions concern
decision-relevant prediction, rather than changes in raw state covariance.

\subsection{An exponential omission curve}

An exponential fresh omission curve gives a global radius solution using
one derivative. This generic calculation will also apply to the spatial
model in the next subsection.

\begin{proposition}[Optimal radius for exponential omission]
\label{prop:generic_exponential_radius}
Under Theorem~\ref{thm:spatiotemporal_composition}, suppose
$\eta_S(r)=\eta_0e^{-2r/\ell_c}$ with $0<\eta_0<1$, $\ell_c>0$,
and $\tau(r)=r/v$ with $v>0$. For continuous $r\ge0$, the unique global
optimal radius is
\begin{equation}
\boxed{r^\star=\frac{\ell_c}{2}
\left[\log\!\left(\eta_0\left(1+\frac{vT}{\ell_c}\right)\right)\right]_+,}
\label{eq:generic_exponential_rstar}
\end{equation}
where $[z]_+=\max\{z,0\}$. It is positive exactly when
$\eta_0(1+vT/\ell_c)>1$. With a continuous radius cap $D$,
$r_D^\star=\min\{D,r^\star\}$.
\end{proposition}
\begin{proof}
Differentiating the captured benefit gives
\[
G'(r)=2e^{-2r/(vT)}\left[
\eta_0\left(\frac{1}{vT}+\frac{1}{\ell_c}\right)e^{-2r/\ell_c}
-\frac{1}{vT}\right].
\]
The bracket decreases strictly to a negative limit. If its value at zero
is nonpositive, $G$ decreases for all positive radii and zero is its unique
maximum, including when the initial derivative is zero. Otherwise $G$
increases up to its unique crossing and decreases thereafter. Solving for
the crossing gives the formula; a cap truncates the maximizing radius.
\end{proof}

The cap in Fig.~\ref{fig:radius-regret} limits observation of an infinite
field, so information beyond $D$ remains omitted. A graph diameter is
different: the fresh view then includes the whole finite graph and omission
is zero. The generic result assumes an omission curve; the next model
derives it from state statistics and the desired action.

\subsection{The canonical three-length result}
\label{subsec:canonical_line_model}

We now specialize to a jointly Gaussian space--time field on the infinite line,
with spatial coordinate $z\in\mathbb R$. Let $\ell_s>0$ be the spatial correlation length, the
distance over which field values remain correlated, and let $T>0$ be the
temporal coherence time. The zero-mean field has covariance
\begin{equation}
\mathbb E[\theta(z,t)\theta(z',s)]
=\sigma^2e^{-|z-z'|/\ell_s}e^{-|t-s|/T}.
\label{eq:canonical_space_time_covariance}
\end{equation}
At a fixed time, this exponential covariance makes the field a spatial
Ornstein--Uhlenbeck process \cite{williams2006gaussian}. Given the field
value at a point, its two sides are conditionally independent; given an
observed interval, each unobserved tail depends only on the adjacent boundary
value. This spatial Markov property is used in Lemma~\ref{prop:canonical_eta_S}.
The decision-relevance length $\ell_c>0$ determines how far remote states
influence the desired local action. Define
\begin{equation}
x^\star(z,t)=\int_{-\infty}^{\infty}w_{\ell_c}(u)\theta(z+u,t)\,du,
\qquad w_{\ell_c}(u)=\frac{1}{2\ell_c}e^{-|u|/\ell_c}.
\label{eq:canonical_target}
\end{equation}
The kernel integrates to one. Each location incurs squared tracking loss
\begin{equation}
f_z(x,\theta)=\frac q2[x(z)-x^\star(z,\theta)]^2,\qquad q>0.
\label{eq:canonical_tracking_cost}
\end{equation}
There is no action coupling, so the optimal policy is the conditional mean.
We measure regret per unit length, equal by stationarity to expected
pointwise loss, with $R_\infty=(q/2)\operatorname{Var}(x^\star(z,t))$.
The radius-$r$ observation consists of all field values within distance $r$:
\begin{equation}
\mathcal O_r(z,t)=(\theta(z+u,t):|u|\le r).
\label{eq:canonical_radius_observation}
\end{equation}
Information propagates at speed $v>0$, giving
\begin{equation}
\tau(r)=r/v,\qquad L_T=vT.
\label{eq:temporal_propagation_length}
\end{equation}
The propagation length $vT$ is the distance information travels in one
coherence time.

Fix a location $z$, a radius $r$, and a delay $\tau$, and write
$a=e^{-\tau/T}$. By the separable covariance in
\eqref{eq:canonical_space_time_covariance}, the innovation
$\theta(y,t)-a\theta(y,t-\tau)$ is uncorrelated with
$\theta(y',t-\tau)$ for every $y,y'$. Joint Gaussianity therefore makes
the innovation field independent of the entire field at time $t-\tau$.
Integrating against the decision kernel in \eqref{eq:canonical_target}
and conditioning on the observed interval gives
\[
\mathbb E[x^\star(z,t)\mid \mathcal O_r(z,t-\tau)]
= a\,\mathbb E[x^\star(z,t-\tau)\mid \mathcal O_r(z,t-\tau)].
\]
Under the pointwise squared loss, the gain over acting without observations
is $q/2$ times the variance of the conditional mean. By stationarity,
the conditional mean at time $t-\tau$ has the same variance as its
fresh-snapshot counterpart at time $t$. The delayed interval therefore
retains $a^2$ times the gain from a fresh interval of the same radius.

\begin{lemma}[Exponential omission by self-similarity]
\label{prop:canonical_eta_S}
For the field \eqref{eq:canonical_space_time_covariance}, decision
\eqref{eq:canonical_target}, and observation
\eqref{eq:canonical_radius_observation}, the normalized zero-delay spatial
omission error is
\begin{equation}
\eta_S(r)=\eta_0e^{-2r/\ell_c},\qquad
\eta_0=\frac{\ell_c}{\ell_c+2\ell_s}.
\label{eq:canonical_eta_S}
\end{equation}
\end{lemma}

\begin{proof}
By stationarity, set $z=0$ and suppress time. The observed middle of the
target integral contributes no uncertainty. Conditional on the two boundary
values, the unobserved tails are independent by the spatial Markov property.
Moving each tail outward by $r$ scales its weight by $e^{-r/\ell_c}$,
while stationarity leaves its conditional variance relative to its boundary
unchanged. Squaring the weight therefore gives
$\eta_S(r)=e^{-2r/\ell_c}\eta_S(0)$.
The covariance integrals in Appendix~\ref{app:technical-canonical} give
$\eta_S(0)=\ell_c/(\ell_c+2\ell_s)$.
\end{proof}

\begin{samepage}
\begin{theorem}[Optimal synchronized-snapshot radius in the canonical model]
\label{thm:canonical_optimal_radius}
Suppose the field obeys Eq.~\eqref{eq:canonical_space_time_covariance},
the target is Eq.~\eqref{eq:canonical_target}, the loss is
Eq.~\eqref{eq:canonical_tracking_cost}, and $\tau(r)=r/v$.
The unique optimal radius on the infinite line is
\begin{equation}
\boxed{r^\star=\frac{\ell_c}{2}
\left[\log\frac{\ell_c+vT}{\ell_c+2\ell_s}\right]_+,}
\label{eq:canonical_optimal_radius}
\end{equation}
where $[z]_+=\max\{z,0\}$. In particular,
\begin{equation}
r^\star>0\quad\Longleftrightarrow\quad vT>2\ell_s.
\label{eq:positive_radius_condition}
\end{equation}
Within the interior regime, $r^\star$:
\begin{enumerate}
    \item increases with the temporal coherence time $T$;
    \item increases with the information-propagation speed $v$;
    \item decreases with the spatial correlation length $\ell_s$; and
    \item increases with the decision-relevance length $\ell_c$.
\end{enumerate}
\end{theorem}
\end{samepage}

\begin{proof}
Lemma~\ref{prop:canonical_eta_S} gives
$\eta_0=\ell_c/(\ell_c+2\ell_s)$. Substitution into
Proposition~\ref{prop:generic_exponential_radius} gives the radius and
the threshold $vT>2\ell_s$. In the positive-radius regime,
$\partial r^\star/\partial(vT)=\ell_c/[2(\ell_c+vT)]>0$ and
$\partial r^\star/\partial\ell_s=-\ell_c/(\ell_c+2\ell_s)<0$.
The integral representation in Appendix~\ref{app:technical-canonical}
proves the increase with $\ell_c$. Here the omission prefactor varies with
the lengths; these comparative statics do not hold automatically for the
generic model with fixed $\eta_0$.
\end{proof}

Equation~\eqref{eq:generic_exponential_rstar} also covers the free-$\eta_0$
curves in Fig.~\ref{fig:radius-regret}; its right panel reaches the widest
allowed radius on an infinite field and still omits information beyond $D$.
The threshold $vT>2\ell_s$ has a direct interpretation: when nearby states
are strongly correlated, the local reading already tells the agent more,
so propagation must be faster before a wider view pays.
In this canonical model, longer state correlations let a local reading predict
farther-away state, whereas broader decision weights make remote states
matter more to the target.

When $\ell_c\gg vT,\ell_s$ and $vT>2\ell_s$,
$r^\star\to(vT-2\ell_s)/2$ as $\ell_c\to\infty$ with $vT,\ell_s$ fixed.
Even as increasingly distant states matter to the decision, propagation
delay limits how far the agent should look.

\section{Numerical Illustrations}
\label{sec:numerical-evaluation}

We ask when fresh local information beats a delayed global view, why an
intermediate snapshot radius can be best, whether changing the decision alone
can change the best radius, and how the framework applies across graph geometries. The
examples use common-rate Gaussian-affine environments.
Appendix~\ref{app:technical-numerics} gives the settings and sampling
procedures needed to reproduce the four figures.

\subsection{When does fresh local information beat a delayed global view?}
\label{subsec:numerical-temporal-crossover}

Stronger decision coupling increases the delay worth tolerating for a global
view. Figure~\ref{fig:local-global-phase} illustrates this in the
shared-resource model of Section~\ref{subsec:coupling_example}.
The exact finite-$n$ boundary is Eq.~\eqref{eq:aggregate_rho_star}; its
large-$n$ limit is $\rho^\star=\tfrac12\log(1+\gamma/q)$
from Eq.~\eqref{eq:large_n_threshold}. Above the boundary, fresh local
information gives lower regret. The finite-size curves approach the
large-system limit as $n$ increases.

We fix $q=1$ and compare $n\in\{5,10,25,100\}$ over nine coupling values.
The Monte Carlo estimates infer the crossover from paired local and
no-information losses sampled at the current time, using
$\widehat\rho^\star=\tfrac12\log[\widehat R_\infty/
(\widehat R_\infty-\widehat R_{\rm loc})]$.
Their agreement with the exact curves checks the loss calculation and the
inferred boundary; the estimates do not sample a separate delayed system
at each age.

\begin{figure}[!htbp]
\centering
\includegraphics[width=0.98\linewidth]{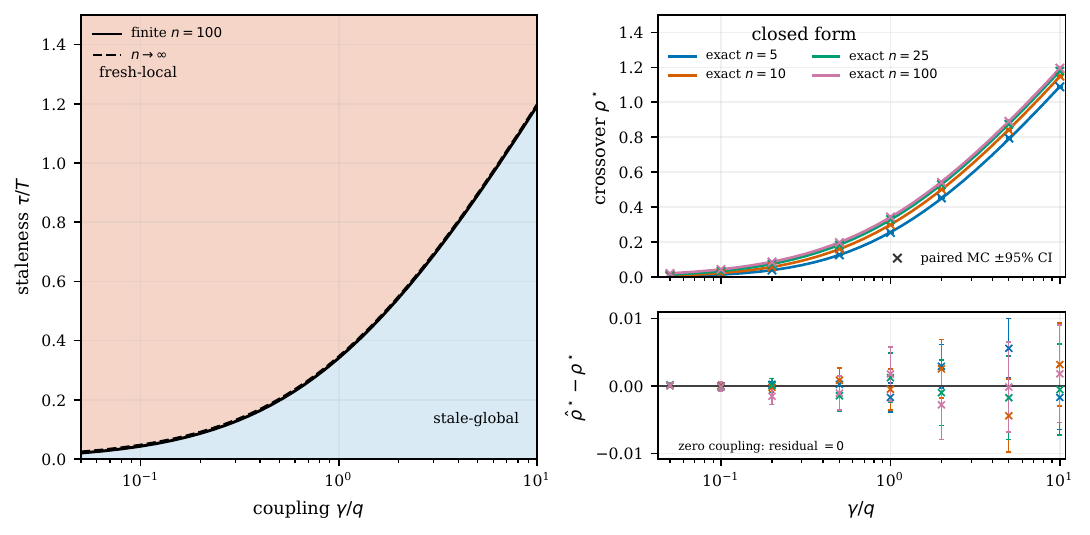}
\caption{Fresh-local versus stale-global crossover in the shared-resource
model. Left: the exact finite-$n=100$ boundary (solid) and large-$n$ limit
(dashed). Upper right: exact boundaries and Monte Carlo estimates inferred
from sampled local and no-information losses for four system sizes, with paired
delta-method 95\% intervals. Lower right: estimated minus exact boundary.
The zero-coupling boundary is zero and is omitted from logarithmic axes.}
\label{fig:local-global-phase}
\end{figure}

\subsection{Why can an intermediate radius be best?}
\label{subsec:numerical-radius-curves}

Figure~\ref{fig:radius-regret} shows three possibilities for a synchronized
snapshot: the best radius is local (left), interior (middle), or as wide as
the imposed cap permits (right). With
$\eta_S(r)=\eta_0e^{-2r/\ell_c}$ and $\tau(r)=r/v$, the temporal loss grows
with radius, while the residual spatial loss falls. Their sum can therefore
be minimized at either endpoint or at an interior balance. For $\ell_c=1$
and cap $D=2$, the three settings give constrained analytical optima
$0$, $0.62638$, and $2$, respectively. The cap limits observation on an
infinite spatial field: even at $D$, the snapshot omits some of that field.

The open markers check this decomposition using independent Gaussian
samples with prescribed scalar component variances. They check variance
addition in Eq.~\eqref{eq:R_decomposition}, rather than simulate a spatial
field. The spatial and temporal curves explain which component limits the
value of widening each snapshot.

\begin{figure}[!htbp]
\centering
\includegraphics[width=\linewidth]{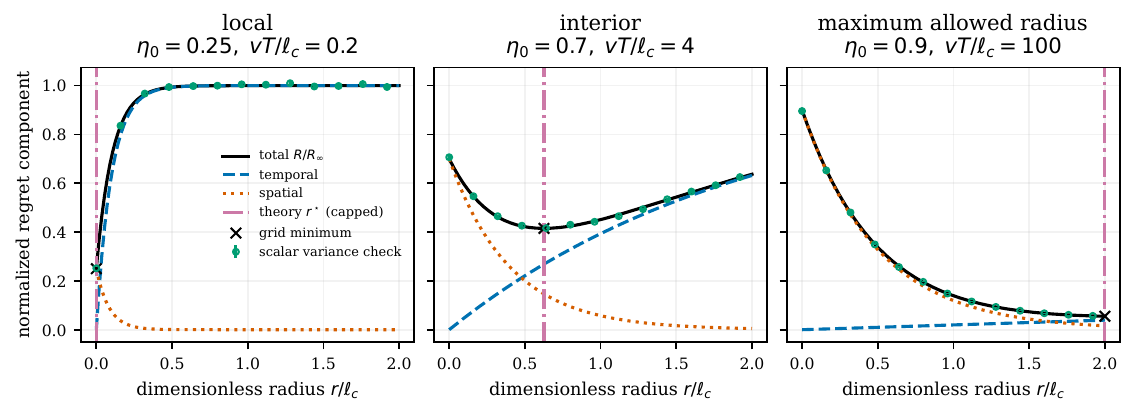}
\caption{Local, interior, and capped optimal radii, from left to right.
Curves show total regret and its temporal and residual-spatial components
under exponential omission and linear delay. Open markers with 95\%
intervals use prescribed scalar variances. Dash-dotted lines and crosses
mark analytical constrained and grid minimizers. At the cap,
$\eta_S(D)>0$: the snapshot still omits part of the infinite field.}
\label{fig:radius-regret}
\end{figure}

\subsection{Can the best radius change on the same network?}
\label{subsec:numerical-decision-relevance}

On the same 64-node ring, changing only the decision operator moves the
optimal radius from 1 to 2 to 5 hops (Fig.~\ref{fig:decision-relevance}).
The graph, state covariance, coherence time $T=6$, and latency $0.12$ per
hop stay fixed. With $Q=I$, we use the row-normalized operators
$K_{ij}=e^{-d_G(i,j)/\ell_c}/\|e^{-d_G(i,\cdot)/\ell_c}\|_2$
for $\ell_c=0.7,2,6$. The fixed covariance is
\begin{equation}
\Sigma_S=\sigma^2(\kappa_s^2I+L_G)^{-\nu},\qquad \kappa_s,\nu>0,
\label{eq:graph_spectral_covariance}
\end{equation}
where $L_G$ is the ring Laplacian and the scale gives average marginal
variance one. This covariance emphasizes low graph-frequency modes
\cite{shuman2013emerging,lindgren2011explicit}.

Exact Gaussian conditioning gives each operator's omission curve, and the
composition law gives regret at every integer radius. Each curve is
normalized by its own $R_\infty=\tfrac12\operatorname{tr}(K\Sigma_SK^\top)$.
The raw state correlation is identical in all three cases, but the value of
remote observations depends on how they affect the decision. State
correlation alone therefore cannot determine the preferred radius.
This ring example illustrates the exact one-hop comparison in
Eq.~\eqref{eq:discrete_radius_improvement}.

\begin{figure}[!htbp]
\centering
\includegraphics[width=0.94\linewidth]{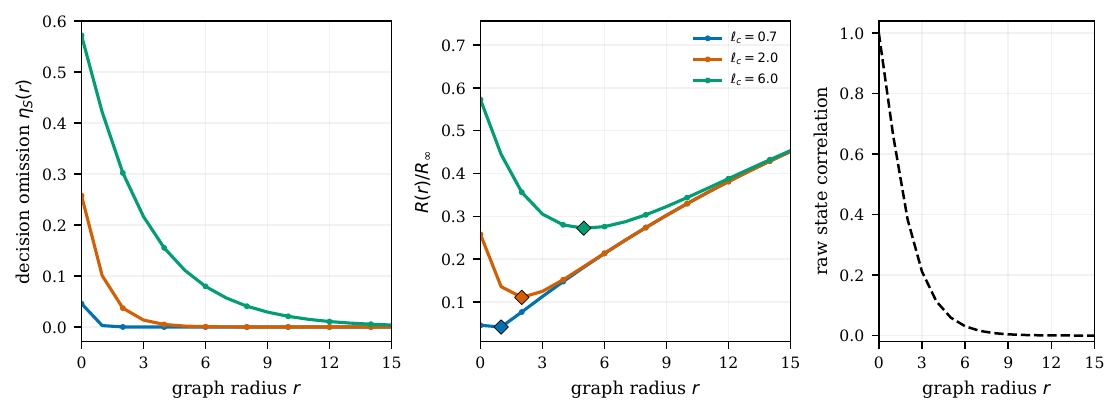}
\caption{Changing the decision on a fixed 64-node ring. Left: decision-weighted
omission for three decision operators. Middle: composed regret, normalized
by each operator's no-information regret, with optimal radii 1, 2, and 5 hops
marked. Right: raw state correlation, identical across operators.
The graph, covariance, temporal coherence, and per-hop latency are fixed.}
\label{fig:decision-relevance}
\end{figure}

\subsection{Does the framework extend beyond the canonical spatial model?}
\label{subsec:numerical-general-graphs}

The closed-form radius in Theorem~\ref{thm:canonical_optimal_radius} assumes
a homogeneous infinite line. To illustrate the more general composition law
of Theorem~\ref{thm:spatiotemporal_composition}, we return to the finite-graph
quadratic model and compare one 64-node path, ring, $8\times8$ grid, and
random geometric graph. All use $Q=I$, the Laplacian covariance in
Eq.~\eqref{eq:graph_spectral_covariance} with $(\kappa_s,\nu)=(0.8,1.5)$,
and average marginal variance one. The reference setting fixes $T=5$,
decision range $\ell_c=2$, and delay $\tau(r)=0.12r$.

For each graph, the row-normalized operator
$K_{ij}\propto e^{-d_G(i,j)/\ell_c}$ defines the desired action.
Gaussian conditional covariance gives each agent's residual action variance
given its fresh neighborhood. Summing these variances and dividing by
$\operatorname{tr}(K\Sigma_S K^\top)$ gives $\eta_S(r)$.
The common-rate composition law then gives delayed regret. We evaluate every
integer radius from zero through the graph diameter and select the minimum,
without assuming unimodality or fitting an exponential omission curve.

Figure~\ref{fig:general-graphs} gives reference optima of 2, 2, 3, and 2 hops,
respectively: all are interior. Across the displayed decision-range sweep,
the path and ring also have local optima at $\ell_c=0.7$; none of the displayed
settings selects the full graph. The sampled coherence-time and decision-range
sweeps weakly increase the preferred radius on these four instances.
These calculations illustrate architecture selection on finite graphs when
no closed-form radius is available. They use synchronized snapshots and
common-rate Gaussian dynamics; they neither extend the canonical radius
formula to these geometries nor establish a new theorem.

\begin{figure}[!htbp]
\centering
\includegraphics[width=\linewidth]{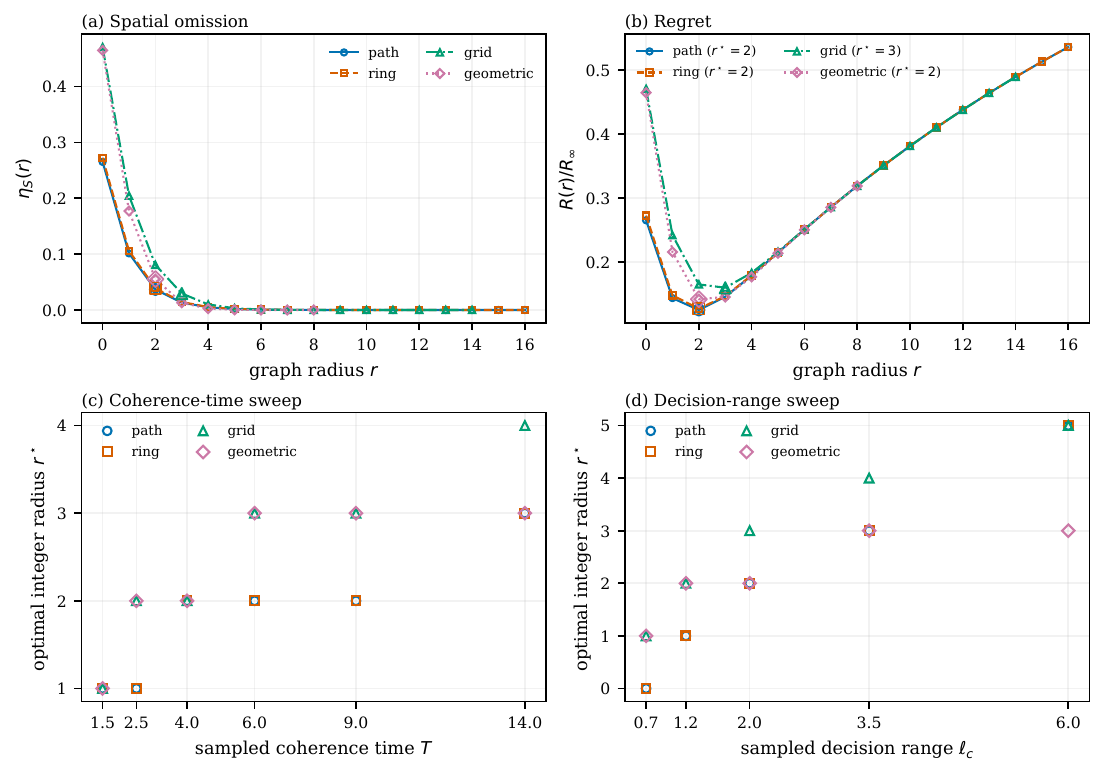}
\caption{Composition-law evaluation on four 64-node graphs. Top: spatial
omission and delayed regret with $Q=I$, $T=5$, $\ell_c=2$, and delay $0.12r$;
curves show integer radii up to 16, or the diameter if smaller, while
optimization uses every radius. Bottom: varying $T$ or $\ell_c$ alone,
holding each graph and covariance fixed. Markers show sampled integer
optima, not transition locations. The canonical radius formula is not used.}
\label{fig:general-graphs}
\end{figure}

\section{Related Work}
\label{sec:related_work}

We organize the closest literatures around three distinctions: iterative
computation versus information architecture, controlled dynamics versus an
exogenous environment, and age-based freshness versus decision-relevant value.
These distinctions frame the more specific connections developed below.

\subsection{Static Team Decision Theory}
\label{subsec:related_static}

Marschak and Radner also optimize the extent of information when more
information takes longer to obtain. Their marginal condition
\cite[Ch.~7, Eq.~(6.5)]{marschak1972economic} gives the balance in
Proposition~\ref{thm:marginal_balance}. For their identical-interaction
example, the group-size-one case of their Eq.~(6.6) gives the local
information share $1-\eta_S(0)$ in
Proposition~\ref{prop:aggregate_local}. Here the design variable is the
radius of a synchronized neighborhood. Computing how much that
neighborhood reveals about the desired action gives the radius formula
on the line and the comparison of decision weights on finite graphs.

\paragraph{Optimal team rules.}
Radner established foundational optimality conditions for static teams, in which decision makers share a common objective but act on different information \cite{radner1962team}. Under appropriate convexity and regularity conditions, the resulting stationarity conditions are sufficient for global team optimality; in the quadratic Gaussian specialization, the optimal decision rules are affine \cite{radner1962team,marschak1972economic}. Lemma~\ref{lem:radner_projection} specializes this framework to the weighted
prediction problem induced by a deterministic quadratic Hessian; it is not a
literal restatement of Radner's theorem. Krainak, Speyer, and Marcus subsequently relaxed the sufficient conditions associated with Radner's stationarity result and extended the theory to an exponential-quadratic criterion with jointly Gaussian state and observations \cite{krainak1982staticI}. In a companion paper, they formulated affine linear-quadratic-Gaussian and linear-exponential-Gaussian team problems as constrained parameter optimizations and represented optimal team gains as projections of the corresponding centralized gains \cite{krainak1982staticII}.

\paragraph{Information as a design variable.}
Classical team theory also treated information itself as an organizational design object. Marschak and Radner analyzed the value of information, delayed information and the tradeoff between timeliness and completeness, and formulated team organization in terms of optimal networks \cite{marschak1972economic}. More recently, Summers, Li, and Kamgarpour formulated information-structure design as the selection of additional measurements or communication links jointly with the corresponding team decision strategies \cite{summers2017information}. Afshari and Mahajan characterized quadratic Gaussian static teams whose observations decompose into information common to all agents and agent-specific local information \cite{afshari2017static}. They subsequently applied this decomposition to team-optimal decentralized filtering over communication graphs with delayed information sharing, where sufficiently old measurements constitute common information while more recent measurements remain local \cite{afshari2018team}. Rusmevichientong and Van Roy studied large teams in which each agent observes the local cost structure only within a graph radius, using centralized performance as a benchmark \cite{rusmevichientong2003decentralized}. Gamarnik, Goldberg, and Weber studied decentralized decisions based on local information in random decision networks and related their near-optimality to a correlation-decay, or long-range-independence, property \cite{gamarnik2014correlation}.

\subsection{Distributed Optimization, Communication Constraints, and Scope}
\label{subsec:related_distributed_control}

Communication structure has also been treated as a design variable. Langbort
and Gupta study interconnection topology when communication is assigned a
topology-dependent cost \cite{langbort2009minimal}, and Matni co-designs
communication-delay structure and $H_2$ control using regularization
\cite{matni2017communication}. In wireless consensus, Vanka, Gupta, and Haenggi
explicitly trade communication radius against interference-induced delay
and show that the preferred connectivity depends on network geometry
\cite{vanka2010power}. Most directly, Ballotta, Jovanovi\'c, and Schenato
parameterize control architectures by communication scope: each agent receives
measurements from nodes within a chosen number of communication hops, with an
architecture-dependent delay that increases with that scope. After optimizing
the feedback gains for each scope, they show that sparse architectures can
outperform all-to-all feedback \cite{ballotta2023latency}. Ballotta and Gupta
obtain an analogous conclusion when maximizing consensus convergence speed
under hop-dependent communication delay \cite{ballotta2023sparser}. These
results establish that scope-dependent delay can itself induce a sparse or
intermediate optimal control architecture.

Distributed optimization studies how agents compute a common solution using local computation and message exchange. Classical asynchronous gradient methods allow delayed interprocessor communication while retaining convergence guarantees under suitable boundedness conditions \cite{tsitsiklis1986distributed}, while distributed subgradient methods establish convergence when agents exchange information locally over time-varying network topologies \cite{nedic2009distributed}. In these canonical formulations, the communication process constrains an iterative computation of an optimizer. Our setting separates that computational question from information architecture: the full-information decision at each epoch is the benchmark, and the object of design is the information available when that decision must be made.

Networked and decentralized control place communication restrictions directly
in the feedback loop, including sampling, packet loss, communication topology,
and delay \cite{hespanha2007survey}. For spatially invariant plants, early work
showed that optimal controllers can exhibit inherent spatial localization
\cite{bamieh2002spatial}, while subsequent work imposed explicit finite-speed
communication structure. Arbelaiz, Bamieh, Hosoi, and Jadbabaie characterize
the spatial decay of optimal estimator gains, using the decay rate as a proxy
for the distance over which measurements are useful in a spatially distributed
dynamic estimation problem \cite{arbelaiz2025estimation}. Voulgaris, Bianchini, and Bamieh studied $H_2$ control
under delayed communication requirements \cite{voulgaris2003optimal}; Bamieh
and Voulgaris formulated distance-dependent information propagation through
funnel causality \cite{bamieh2005convex}; and Fardad and Jovanovi\'c developed
state-space models for distributed controllers with finite communication speed
\cite{fardad2011finite}. These works connect spatial scope and propagation time to dynamic feedback
for controlled plants.

Our setting differs in the decision problem: an exogenous-state static team,
with the current full-information action as the benchmark and an architecture
valued by the prediction error it induces for that action. 
Ballotta, Arbelaiz, Gupta, Schenato, and Jovanovi\'c study a different but
complementary question for spatially invariant dynamic plants: how delayed
measurements from different spatial locations alter the spatial locality of the
optimal feedback kernel \cite{ballotta2026role}. Their results operate through
closed-loop plant dynamics and controller localization, whereas ours operate
through prediction of an exogenously varying full-information decision.

\subsection{Freshness and Age of Information}
\label{subsec:related_aoi}

Jiang and Zhou \cite{jiang2019status} study a one-dimensional
Gauss--Markov field with exponential correlation in space and time.
They choose the sampling density to minimize remote field-estimation
error when the density also affects queueing delay.
Munari et al.~\cite{munari2025spatiotemporal} study how spatial
correlation and update freshness affect the choice of communication
range for remote monitoring. Here the design variable is the radius
of a synchronized snapshot. Its value depends on how strongly the
states inside that radius predict the agent's desired action.

Age of Information (AoI) formalizes the timeliness of status information through the age of the freshest received update, $\Delta(t)=t-u(t)$. Seminal work showed that minimizing AoI is distinct from maximizing throughput or minimizing packet delay and initiated a large literature on update generation, queueing, and scheduling \cite{kaul2012realtime,yates2021aoi}. Subsequent work generalized linear age to application-specific nonlinear age penalties \cite{sun2017update}. Connections to remote estimation make the distinction between age and task loss explicit: for a Gauss--Markov Ornstein--Uhlenbeck process, estimation error under signal-agnostic sampling is a nonlinear function of age, whereas signal-aware sampling can exploit the realized estimation error and outperform age-optimal policies \cite{ornee2021ou}.

More recent work has pushed freshness metrics toward task-aware objectives. Shisher et al. showed that remote-inference loss can be nonmonotonic in age and designed communication policies around inference performance \cite{shisher2024timely}; subsequent work considers correlated sources for which inference penalties depend jointly on multiple sources' ages \cite{shisher2026correlated}. Li et al. go further toward decision-aware freshness by jointly optimizing sampling and remote decision making under stale observations \cite{li2026freshness}.

\subsection{Centralization, Delegation, and Organizational Economics}

A closely related organizational-economics literature asks how dispersed information, communication constraints, and decision rights determine organizational form. Aoki compares hierarchical and horizontal information structures when local conditions are uncertain and central monitoring or response is limited \cite{aoki1986horizontal}. In the information-processing tradition, Radner studies organizations of limited-capacity processors trading off processing resources against decision delay \cite{radner1993organization}; Bolton and Dewatripont derive communication networks from processing and communication costs \cite{bolton1994firm}; and Van Zandt studies real-time decentralized processing in which computational delay constrains the use of recent information \cite{vanzandt1999realtime}. Dessein and Santos likewise study adaptation to local information versus coordination under imperfect communication \cite{dessein2006adaptive}.

A complementary delegation literature makes decision rights endogenous when information and objectives are dispersed. Aghion and Tirole distinguish formal from real authority \cite{aghion1997formal}, while Dessein studies delegation as an alternative to strategic communication \cite{dessein2002authority}. Most closely related to our centralization framing, Alonso, Dessein, and Matouschek compare centralized and decentralized coordination when managers privately observe local conditions \cite{alonso2008coordination}. Our model abstracts from incentive conflicts and endogenous information acquisition: agents share a common objective and the environment evolves exogenously. The tradeoff studied here is spatio-temporal: broader synchronized snapshots can improve coordination but arrive older. 

\subsection{Spatial Statistics and Graph Signal Processing}
\label{subsec:related_spatial}

Spatial statistics provides the closest statistical antecedent to our spatial-predictability model. Kriging casts spatial prediction as an optimal linear prediction problem under spatial dependence \cite{cressie1990origins}. For lattice data, Besag's conditional formulation represents spatial dependence through local Markov neighborhoods \cite{besag1974spatial}; Lindgren, Rue, and Lindstr\"om later connect Mat\'ern Gaussian fields to sparse Gaussian Markov random fields through stochastic partial differential equation representations \cite{lindgren2011explicit}. Space--time statistics treats separability as an optional modeling restriction: Gneiting constructs broad classes of stationary nonseparable covariance functions \cite{gneiting2002nonseparable}. These works provide the statistical context for our covariance-based spatial model and its separable and nonseparable distinction.

Graph signal processing extends analogous ideas to irregular domains. Shuman et al.\ use graph-Laplacian eigenvectors to define a graph Fourier domain and associate low graph frequencies with signals that vary slowly over the graph \cite{shuman2013emerging}. Graph-sampling theory asks when bandlimited graph signals can be exactly recovered from partial observations \cite{chen2015discrete}, while stationary graph signal processing gives spectral descriptions and Wiener-type estimators for random graph signals \cite{perraudin2017stationary}. Here the prediction target is the full-information action, so a state pattern's
importance depends on the decision operator as well as on its spatial smoothness.

\section{Discussion and Conclusions}
\label{sec:discussion_conclusions}

\paragraph{What the examples show.}
On the same 64-node ring in Fig.~\ref{fig:decision-relevance}, changing only
the decision target moves the preferred radius from 1 to 2 to 5 hops.
The graph, state covariance, coherence time, and per-hop latency stay fixed.
Graph distance determines what it takes to acquire information; the decision
determines whether that information is useful. This suggests selecting remote
observations for their effect on the action, rather than expanding scope
solely because remote states remain correlated.
Figure~\ref{fig:general-graphs} also illustrates the common-rate composition
law on finite path, ring, grid, and geometric graphs beyond the canonical
spatial model, with optimal radii found by exhaustive comparison.

\paragraph{Scope and limitations.}
The environment evolves exogenously: actions do not change its future state
or the information available to other agents. When actions affect later
observations, they may also signal information;
the dynamic-team extension is discussed below. The principal family uses
synchronized snapshots, while mixed-age and history-based architectures can
retain information that these snapshots discard. Stationarity makes the
expected snapshot regret time-invariant; Appendix~\ref{app:technical-projection}
distinguishes this from fixed-epoch optimization and sample-path averages.

The projection identity assumes quadratic costs, a deterministic positive-
definite Hessian, and unconstrained Euclidean actions. The explicit stale-global
and composition laws additionally require affine full-information decisions
and a stationary common-rate innovation structure with conditionally
mean-zero innovations. The composition law allows coupled actions through any fixed $Q\succ0$
and applies to non-Gaussian states. The closed-form radius in
Theorem~\ref{thm:canonical_optimal_radius} uses an infinite line, a jointly
Gaussian field with the separable exponential covariance in
\eqref{eq:canonical_space_time_covariance}, an exponential decision
kernel, uncoupled squared loss, and propagation delay $\tau(r)=r/v$.
The finite-graph examples in
Sections~\ref{subsec:numerical-decision-relevance}
and~\ref{subsec:numerical-general-graphs} use Gaussian states and $Q=I$,
so each agent predicts its own weighted target.
Nonquadratic or discrete decisions, state-dependent Hessians, and feasibility
constraints such as nonnegativity, capacity, or simplex restrictions require
different analytical tools. Outside the stated assumptions, one should work
directly with decision-relevant prediction-error functions or suitable bounds
rather than assume the same factorization.

The numerical illustrations specialize the framework to Gaussian models.
Heterogeneous temporal rates, nonseparable covariance, nonlinear desired
actions, and mixed-age architectures remain directions for further study.

When spatial patterns change at different rates, a single coherence time can
overstate or understate the value of delayed information, depending on which
patterns drive the decision.

\paragraph{Most important extensions.}
A first extension is to price information acquisition explicitly. The present
model represents communication and sensing constraints mainly through the
feasible scope--delay relationship. Adding an acquisition cost or budget could
produce sparse, nonuniform topologies and compressed summaries that preserve
decision-relevant modes with less latency.

A second extension is adaptation to unknown or nonstationary statistics.
Temporal coherence, spatial predictability, and decision sensitivity may need
to be estimated online while the architecture changes on a slower timescale
than operational decisions. The challenge is to account for estimation
uncertainty while retaining a stable architecture-selection rule.

The most important extension is action-dependent state evolution and dynamic
teams, where actions alter future state and may signal information to other
agents. Extending information-architecture design beyond the exogenous-state
baseline is therefore not a routine modification.
Witsenhausen's counterexample \cite{witsenhausen1968counterexample} explains
why: once control and signaling interact, the fixed-subspace projection
characterization generally no longer solves the team problem.

\paragraph{Design implication.}
Under the single-crossing condition, an interior optimal radius balances
the marginal gain in decision information against the loss from delay.
If there is no interior crossing, the optimum lies at a boundary.
On a general finite graph,
compare composed regret over all admissible radii.

\paragraph{AI Use Statement.}
The development of this paper made substantial use of AI tools, including
Claude Code, ChatGPT/Codex, Gemini, and Grok, for drafting and revising text,
developing code and figures, mathematical formulation and proofs.
The human authors formulated the original problem, made substantial
contributions throughout the paper's development, reviewed the manuscript, and
accept full responsibility for its contents.

\paragraph{Code and Data Availability.}
Appendices~\ref{app:technical-aggregate}--\ref{app:technical-numerics} provide detailed algebraic
derivations and numerical reproducibility information;
Appendix~\ref{app:technical-projection} contains the projection proof and
coupled optimality equations. A separate mathematical
companion in the repository provides background on the projection and
stochastic-process tools used here.
The manuscript source, numerical code, experiment configurations, processed
reference results, and figure-generation pipeline are available in the public
TINA repository at \url{https://github.com/ANRGUSC/TINA}
\cite{tina2026software}.

\bibliographystyle{IEEEtran}
\bibliography{references}

\appendix
\numberwithin{equation}{section}
\input{technical-appendices-v2.tex}

\end{document}

%% file: technical-appendices-v2.tex
\section{General formulation and fixed-architecture optimum}
\label{app:technical-projection}
\label{subsec:projection}

This appendix preserves the broader fixed-epoch quadratic formulation.
Its target may be non-affine, unlike the temporal composition theorem.

On a probability space $(\Omega,\mathcal F,\mathbb P)$, let
$\theta=\{\theta_t\}_{t\in\mathbb T}$ be an environment process,
with discrete or continuous time. Write
\begin{equation}
\theta_t=(\theta_{1,t},\ldots,\theta_{n,t})\in\mathbb R^m,
\quad \theta_{i,t}\in\mathbb R^{m_i},\quad \sum_i m_i=m,
\quad \mathbb E\|\theta_t\|^2<\infty.
\label{eq:state_partition}
\end{equation}
Write $\bar\theta=\mathbb E[\theta_t]$ and
$\Sigma_S=\operatorname{Cov}(\theta_t)$ for its mean and covariance at the fixed epoch.
The local state $\theta_{i,t}$ may describe channel conditions, demand,
workload, or sensing quality. This fixed-epoch formulation requires no
stationarity, Gaussianity, independence, or Markov assumptions. Agent $i$ chooses
$x_{i,t}\in\mathbb R^{d_i}$; the joint action is
$x_t=(x_{1,t},\ldots,x_{n,t})\in\mathbb R^d$, with $d=\sum_i d_i$.

The agents share a common per-stage cost
\begin{equation}
    f(x,\theta)
    =
    \frac{1}{2}x^\top Qx
    -
    b(\theta)^\top x
    +
    c(\theta),
    \label{eq:quadratic_cost}
\end{equation}
where $Q\in\mathbb{R}^{d\times d}$ is symmetric and positive definite,
$b:\mathbb{R}^{m}\rightarrow\mathbb{R}^{d}$ is measurable, and
$c:\mathbb{R}^{m}\rightarrow\mathbb{R}$ is an additive state-dependent term.
We assume $\mathbb E\|b(\theta_t)\|^2<\infty$ and
$\mathbb E|c(\theta_t)|<\infty$. The off-diagonal blocks of $Q$ couple agents'
actions, while $b(\theta)$ determines which states affect the preferred action.
This decision coupling need not coincide with proximity in $G$.

If the complete current state $\theta_t$ is available without information
constraints, the unique optimal joint action is
\begin{equation}
    x_t^\star
    \triangleq
    x^\star(\theta_t)
    =
    Q^{-1}b(\theta_t).
    \label{eq:clairvoyant}
\end{equation}
We call $x_t^\star$ the \emph{clairvoyant} or \emph{full-information}
decision.

Completing the square in \eqref{eq:quadratic_cost} gives
\begin{equation}
    f(x,\theta)-f(x^\star(\theta),\theta)
    =
    \frac{1}{2}
    \left(x-x^\star(\theta)\right)^\top
    Q
    \left(x-x^\star(\theta)\right).
    \label{eq:complete_square}
\end{equation}
Hence, for the quadratic model, excess decision cost is exactly a weighted
squared error in reproducing the full-information action.

For the affine specializations used to obtain closed forms, we write
\begin{equation}
    b(\theta)=B\theta+b_0,
    \qquad
    B\in\mathbb{R}^{d\times m},
    \qquad
    b_0\in\mathbb{R}^{d},
    \qquad
    K\triangleq Q^{-1}B\in\mathbb{R}^{d\times m}.
    \label{eq:affine_map_and_sensitivity}
\end{equation}
The matrix $B$ maps state into the linear term of the objective, while $K$ is
the resulting decision-sensitivity operator. The general projection results do
not require this affine form.

\paragraph{Centering the affine model.}
If $b(\theta)=B\theta+b_0$, write
$\theta^c=\theta-\mathbb E[\theta]$ and
$x^c=x-\mathbb E[x^\star]$. Then the centered desired action is
$x^{\star,c}=K\theta^c$, with $K$ as in
Eq.~\eqref{eq:affine_map_and_sensitivity}. Translating each observation and
each action by its known mean gives a bijection of the admissible rules
and leaves $x-x^\star$ unchanged. Thus Eq.~\eqref{eq:complete_square}
is precisely the tracking loss of the main text in centered coordinates.
For a non-affine $b$, the projection result below still holds, but the scalar
state prediction assumption need not give scalar prediction of $x^\star$.

\paragraph{Information and admissible rules.}
Let $\mathcal G_i\subseteq\mathcal F$ be the sigma-algebra generated by
agent $i$'s observations; for example, $\mathcal G_i=\sigma(Z_i)$.
Measurability here means that $u_i$ can be computed from those observations.
An architecture and its policy space are
\begin{equation}
\mathcal A=(\mathcal G_1,\ldots,\mathcal G_n),\qquad
\mathcal S(\mathcal A)=\{u:\mathbb E\|u\|^2<\infty,
\ u_i\text{ is }\mathcal G_i\text{-measurable for every }i\}.
\label{eq:policy_subspace}
\end{equation}
In the local example $\mathcal G^{\rm loc}_{i,t}=\sigma(\theta_{i,t})$.
If observations are enlarged without discarding old ones, then
\begin{equation}
\mathcal G_i\subseteq\mathcal G'_i\ (\text{all }i)
\quad\Longrightarrow\quad
\mathcal S(\mathcal A)\subseteq\mathcal S(\mathcal A'),\quad
R(\mathcal A')\le R(\mathcal A).
\label{eq:info_inclusion}
\end{equation}
This formalizes the information-inclusion argument in the body.

Let
\begin{equation}
    \mathcal{H}
    =
    L^2(
        \Omega,\mathcal{F},\mathbb{P};
        \mathbb{R}^{d}
    )
    \label{eq:hilbert_space}
\end{equation}
and equip $\mathcal{H}$ with the weighted inner product
\begin{equation}
    \langle u,v\rangle_Q
    \triangleq
    \mathbb{E}
    \left[
        u^\top Qv
    \right],
    \qquad
    \|u\|_Q^2
    \triangleq
    \langle u,u\rangle_Q.
    \label{eq:q_inner_product}
\end{equation}
Since $Q\succ0$, the induced norm is equivalent to the ordinary $L^2$ norm.

\begin{lemma}[Fixed-architecture quadratic projection; Radner-type specialization]
\label{lem:radner_projection}
For any static information architecture $\mathcal{A}$, the admissible policy set
$\mathcal{S}(\mathcal{A})$ defined in \eqref{eq:policy_subspace} is a closed
linear subspace of $\mathcal{H}$. Under the quadratic cost
\eqref{eq:quadratic_cost}, the unique optimal policy, up to almost-sure equality, implementable under
$\mathcal{A}$ is
\begin{equation}
    u_{\mathcal{A}}^\star
    =
    \Pi_{\mathcal{S}(\mathcal{A})}
    x_t^\star,
    \label{eq:optimal_projection}
\end{equation}
where $\Pi_{\mathcal{S}(\mathcal{A})}$ denotes orthogonal projection with
respect to \eqref{eq:q_inner_product}. Consequently,
\begin{equation}
    \boxed{
    R(\mathcal{A})
    =
    \frac{1}{2}
    \left\|
        x_t^\star
        -
        \Pi_{\mathcal{S}(\mathcal{A})}
        x_t^\star
    \right\|_Q^2 .
    }
    \label{eq:projection_regret}
\end{equation}
This identity specializes the classical static-team formulation
\cite{radner1962team} to the quadratic model used here.
\end{lemma}

\paragraph{Proof of Lemma~\ref{lem:radner_projection}.}
\begin{proof}
Linearity of $\mathcal{S}(\mathcal{A})$ follows because
$\mathcal{G}_i$-measurability is preserved under linear combinations.
Closedness follows because an $L^2$ limit of
$\mathcal{G}_i$-measurable random variables admits a
$\mathcal{G}_i$-measurable version.

For any admissible policy $u$, \eqref{eq:complete_square} gives
\begin{align}
    \mathbb{E}
    \left[
        f(u,\theta_t)
        -
        f(x_t^\star,\theta_t)
    \right]
    &=
    \frac{1}{2}
    \mathbb{E}
    \left[
        (u-x_t^\star)^\top
        Q
        (u-x_t^\star)
    \right]
    \nonumber\\
    &=
    \frac{1}{2}
    \|u-x_t^\star\|_Q^2.
    \label{eq:projection_proof}
\end{align}
Minimizing this expression over the closed subspace
$\mathcal{S}(\mathcal{A})$ is the standard Hilbert-space projection problem,
yielding \eqref{eq:optimal_projection} and \eqref{eq:projection_regret}.
\end{proof}

The projection error is $Q$-orthogonal to every admissible policy. In particular,
$2R(\mathcal A)=\|x_t^\star\|_Q^2-\|u_{\mathcal A}^\star\|_Q^2$,
which is often useful for computing regret.

\begin{corollary}[Coupled conditional normal equations]
\label{cor:conditional_normal_equations}
The unique team optimum in Lemma~\ref{lem:radner_projection} satisfies
\begin{equation}
    \boxed{
    \mathbb E\!\left[
      \bigl(Q u_{\mathcal A}^{\star}-b(\theta_t)\bigr)_i
      \;\middle|\;
      \mathcal G_i
    \right]
    =0,
    \qquad i=1,\ldots,n.
    }
    \label{eq:conditional_normal_equations}
\end{equation}
With $Q$ partitioned into action blocks, this is equivalently
\begin{equation}
Q_{ii}u_i^\star
+
\sum_{j\ne i}Q_{ij}\,
\mathbb E[u_j^\star\mid\mathcal G_i]
=
\mathbb E[b_i(\theta_t)\mid\mathcal G_i].
\label{eq:block_conditional_normal_equations}
\end{equation}
These are the coupled person-by-person stationarity equations associated with
the classical static-team formulation \cite{radner1962team}. Strict convexity
makes them sufficient for the unique team optimum here. With heterogeneous
information and
off-diagonal $Q$, one generally does not have
$u_i^\star=\mathbb E[x_i^\star\mid\mathcal G_i]$. When all components share a
common information sigma-algebra, deterministic $Q$ makes the joint solution
the ordinary vector conditional expectation.
\end{corollary}

\begin{proof}[Justification of the conditional normal equations]
For an admissible perturbation $h$, differentiating expected quadratic cost
at the optimum gives $\mathbb E[h^\top(Qu^\star-b(\theta_t))]=0$.
Choose perturbations supported on one agent's action block, with that block
an arbitrary square-integrable function of its observations. Conditional
expectation then gives Eq.~\eqref{eq:conditional_normal_equations}.
Conversely, these conditional equations make the linear term vanish in
\[
\mathbb E[f(u^\star+h,\theta_t)-f(u^\star,\theta_t)]
=\mathbb E[h^\top(Qu^\star-b(\theta_t))]+\tfrac12\mathbb E[h^\top Qh].
\]
Positive definiteness makes the remaining term positive for any nonzero
$L^2$ perturbation, proving sufficiency and uniqueness up to almost-sure equality.
\end{proof}

\paragraph{Common-information optimality.}
If every agent observes the same $Z$, let $m=\mathbb E[x^\star\mid Z]$.
For every admissible perturbation $h$, conditioning gives
$\mathbb E[h^\top Q(x^\star-m)]=0$. The error expansion then proves that
$m$ is the unique optimum, regardless of the off-diagonal entries of $Q$.
For different observations this argument does not apply componentwise.
\paragraph{Delayed common snapshots.}
A common old snapshot gives regret
\[
\tfrac12\mathbb E[(x_t^\star-\mathbb E[x_t^\star\mid\theta_{t-\tau}])^\top
Q(x_t^\star-\mathbb E[x_t^\star\mid\theta_{t-\tau}])]
\]
even for a non-affine target. If its conditional mean converges in mean square to the unconditional
mean, this regret converges to the no-information value. Under a stationary
Markov environment, data processing makes this prediction loss nondecreasing
in delay; arbitrary snapshot information sets themselves need not be nested.
\paragraph{Degenerate normalization.}
The best constant is $\mathbb E[x^\star]$; its regret is half the weighted
variance of $x^\star$. In centered coordinates, $R_\infty=0$ and $Q\succ0$
imply $x_t^\star=0$ almost surely. The zero rule has zero regret under
every architecture, so the unnormalized composition law remains valid
while normalized shares are undefined.

Stationarity makes expected snapshot regret time-invariant. It is not needed
for the fixed-epoch projection result. Identifying long-run sample-path
averages with these expectations additionally requires ergodicity.

\subsection{Temporal special cases}
\label{app:temporal-special-cases}
The stationary Gaussian Ornstein--Uhlenbeck process is a sufficient example
of the main conditional-mean assumption, rather than its definition:
\begin{equation}
d\theta_t=-\frac1T(\theta_t-\bar\theta)\,dt
+\sqrt{2/T}\,\Sigma_S^{1/2}dW_t.
\label{eq:OU_process}
\end{equation}
Here $W_t$ is standard multivariate Brownian motion and $\Sigma_S\succeq0$.
Its centered transition has independent Gaussian innovations, and
\begin{equation}
\operatorname{Cov}(\theta_t,\theta_s)=e^{-|t-s|/T}\Sigma_S.
\label{eq:separable_space_time_covariance}
\end{equation}
For a stationary AR(1) process with coefficient $0<c<1$ and independent,
mean-zero square-integrable innovations, the lag-$k$ prediction coefficient
is $c^k$ and the composition weights are $1-c^{2k}$ and $c^{2k}$.
Neither statement asserts that a latest neighborhood snapshot is sufficient
relative to its entire observation history.

\section{Aggregate-coupling calculations}
\label{app:technical-aggregate}
This section supplies the matrix arithmetic for Proposition~\ref{prop:aggregate_local}
above. The individual regrets are
\begin{align}
    R_\infty
    &=
    \frac{
        n\sigma^2
        \left(
            q+\gamma(n-1)
        \right)
    }{
        2q(q+\gamma n)
    },
    \label{eq:aggregate_Rinf}
    \\
    R_{\mathrm{loc}}
    &=
    \frac{
        n\sigma^2\gamma^2(n-1)
    }{
        2q(q+\gamma)(q+\gamma n)
    }.
    \label{eq:aggregate_Rloc}
\end{align}
For $Q=qI+\gamma\mathbf1\mathbf1^\top$,
\[
Q^{-1}=\frac1q\left(I-\frac{\gamma}{q+\gamma n}\mathbf1\mathbf1^\top\right).
\]
\begin{proof}
Recall the coupled normal equations \eqref{eq:block_conditional_normal_equations}
from Appendix~\ref{app:technical-projection}:
\[
 Q_{ii}u_i+\sum_{j\ne i}Q_{ij}\mathbb E[u_j\mid\mathcal G_i]
 =\mathbb E[b_i(\theta_t)\mid\mathcal G_i].
\]
Here $Q_{ii}=q+\gamma$, $Q_{ij}=\gamma$ for $j\ne i$, and $b_i(\theta)=\theta_i$.
Let $\mu_j=\mathbb E[u_j]$. For $j\ne i$, the policy $u_j$ is a function of $\theta_{j,t}$,
which is independent of $\theta_{i,t}$, so
\[
\mathbb E[u_j\mid\mathcal G^{\mathrm{loc}}_{i,t}]
=\mathbb E[u_j]=\mu_j,
\]
while $\theta_{i,t}$ is
$\mathcal G^{\mathrm{loc}}_{i,t}$-measurable. The conditional normal equation
\eqref{eq:block_conditional_normal_equations} therefore gives
\[
(q+\gamma)u_i
+\gamma\sum_{j\ne i}\mu_j
-\theta_{i,t}=0.
\]
Taking expectations yields
\[
q\mu_i+\gamma\sum_j\mu_j=0.
\]
Summing over $i$ implies $\sum_j\mu_j=0$, hence every $\mu_i=0$ and
$u_i=\theta_{i,t}/(q+\gamma)$, which proves
\eqref{eq:optimal_local_shared}.

\emph{No-information regret.} Here $B=I$ and $\Sigma_S=\sigma^2I$, so by \eqref{eq:Rinf_general}, $R_\infty=\tfrac{\sigma^2}2\operatorname{tr}(Q^{-1})$. From the Sherman--Morrison formula above,

\[
\operatorname{tr}(Q^{-1})=\frac1q\Big(n-\frac{\gamma n}{q+\gamma n}\Big)=\frac{n\,(q+\gamma(n-1))}{q\,(q+\gamma n)},
\]

which gives \eqref{eq:aggregate_Rinf}.

\emph{Local regret.} By the Pythagorean form of Lemma~\ref{lem:radner_projection} (the orthogonality identity in Appendix~\ref{app:technical-projection}, with both $x^\star$ and $u^{\rm loc}$ mean zero), $2R_{\rm loc}=\|x^\star\|_Q^2-\|u^{\rm loc}\|_Q^2=2R_\infty-\|u^{\rm loc}\|_Q^2$. Since $u^{\rm loc}=\theta/(q+\gamma)$ and $\operatorname{tr}(Q)=n(q+\gamma)$,

\[
\|u^{\rm loc}\|_Q^2=\frac{\sigma^2\operatorname{tr}(Q)}{(q+\gamma)^2}=\frac{n\sigma^2}{q+\gamma}.
\]

Therefore

\[
R_{\rm loc}=\frac{n\sigma^2}{2}\left[\frac{q+\gamma(n-1)}{q(q+\gamma n)}-\frac{1}{q+\gamma}\right]=\frac{n\sigma^2}{2}\cdot\frac{(q+\gamma(n-1))(q+\gamma)-q(q+\gamma n)}{q(q+\gamma)(q+\gamma n)} .
\]

The numerator reduces to $\gamma^2(n-1)$, which gives \eqref{eq:aggregate_Rloc}.

\emph{Share.} Dividing \eqref{eq:aggregate_Rloc} by \eqref{eq:aggregate_Rinf}, the factors $n\sigma^2/2$ and $q(q+\gamma n)$ cancel, leaving \eqref{eq:aggregate_eta}.

\end{proof}

\section{Canonical spatial omission calculation}
\label{app:technical-canonical}
We expand the proof of Lemma~\ref{prop:canonical_eta_S} above,
including the equal-inverse-length case in the variance integral.
We use a continuous version of the field; observations on an interval can
equivalently be generated by a countable dense subset (including the center
at radius zero).
\begin{proof}
By spatial stationarity, take $z=0$ and suppress the time argument. Write $\beta_c\triangleq1/\ell_c$ and $\beta_s\triangleq1/\ell_s$ for the inverse lengths, so $w(u)=\tfrac{\beta_c}{2}e^{-\beta_c|u|}$ and $\operatorname{Cov}(\theta(u),\theta(u'))=\sigma^2e^{-\beta_s|u-u'|}$.

\textbf{Step 0 (omission is a variance ratio).} With $Q=qI$ the optimal policy given $\mathcal O_r$ is $\mathbb E[x^\star(0)\mid\mathcal O_r]$. Its expected loss is $\tfrac q2\,\mathbb E\big[\operatorname{Var}(x^\star(0)\mid\mathcal O_r)\big]$. For a Gaussian field the conditional variance is a deterministic number: it does not depend on the observed values (the Gaussian conditional covariance is independent of the observed values). Since $R_\infty=\tfrac q2\operatorname{Var}(x^\star(0))$,

\[
\eta_S(r)=\frac{\operatorname{Var}(x^\star(0)\mid\mathcal O_r)}{\operatorname{Var}(x^\star(0))}.
\]

\textbf{Step 1 (split into observed middle and two tails).} Write $x^\star(0)=M_r+T_r^++T_r^-$, where

\[
M_r=\int_{-r}^{r}w(u)\theta(u)\,du,\qquad T_r^+=\int_r^\infty w(u)\theta(u)\,du,\qquad T_r^-=\int_{-\infty}^{-r}w(u)\theta(u)\,du .
\]

The middle $M_r$ is computed from observed values, so it contributes no conditional variance.

\textbf{Step 2 (rewrite the right tail).} Substituting $u=r+s$,

\[
T_r^+=\frac{\beta_c}{2}e^{-\beta_c r}\,Y_r,\qquad Y_r\triangleq\int_0^\infty e^{-\beta_c s}\,\theta(r+s)\,ds .
\]

\textbf{Step 3 (spatial Markov property).} By the exponential-covariance model in Section~\ref{subsec:canonical_line_model}, the field is Markov in $z$. Given the observed segment $[-r,r]$, the right tail $\{\theta(u):u>r\}$ and the left tail $\{\theta(u):u<-r\}$ are conditionally independent, and each depends on the segment only through its adjacent endpoint $\theta(r)$ or $\theta(-r)$. Therefore

\[
\operatorname{Var}(T_r^+\mid\mathcal O_r)=\operatorname{Var}(T_r^+\mid\theta(r)),\qquad \operatorname{Var}(T_r^-\mid\mathcal O_r)=\operatorname{Var}(T_r^-\mid\theta(-r)),\qquad \operatorname{Cov}(T_r^+,T_r^-\mid\mathcal O_r)=0 .
\]

\textbf{Step 4 (the tail variance does not depend on $r$).} Define $C\triangleq\operatorname{Var}(Y_0\mid\theta(0))$. Shifting the field by $r$ does not change its law, so $\operatorname{Var}(Y_r\mid\theta(r))=C$ for every $r$. By reflection symmetry (the covariance depends only on $|u-u'|$), the left tail gives the same value. Hence

\[
\operatorname{Var}(x^\star(0)\mid\mathcal O_r)=2\Big(\frac{\beta_c}{2}\Big)^2e^{-2\beta_c r}\,C=\frac{\beta_c^2}{2}e^{-2\beta_c r}\,C .
\]

At $r=0$ the same argument applies with $\mathcal O_0=\sigma(\theta(0))$ and $M_0=0$, giving $\operatorname{Var}(x^\star(0)\mid\theta(0))=\tfrac{\beta_c^2}{2}C$. Dividing, $\eta_S(r)=\eta_S(0)\,e^{-2\beta_c r}=\eta_S(0)\,e^{-2r/\ell_c}$. The unknown constant $C$ cancels, so it never has to be computed.

\textbf{Step 5 (the constant $\eta_S(0)$).} Conditioning on the single Gaussian variable $\theta(0)$ gives $\operatorname{Var}(x^\star(0)\mid\theta(0))=\operatorname{Var}(x^\star(0))\big(1-\operatorname{Corr}(x^\star(0),\theta(0))^2\big)$ (the Gaussian conditional covariance is independent of the observed values), so $\eta_S(0)=1-\operatorname{Corr}^2$. Two moments are needed.

\emph{Covariance with the local value:}

\[
\operatorname{Cov}(x^\star(0),\theta(0))=\int_{\mathbb R}w(u)\,\sigma^2e^{-\beta_s|u|}\,du=\frac{\sigma^2\beta_c}{2}\int_{\mathbb R}e^{-(\beta_c+\beta_s)|u|}\,du=\frac{\sigma^2\beta_c}{\beta_c+\beta_s}.
\]

\emph{Variance of the target:} $\operatorname{Var}(x^\star(0))=\frac{\sigma^2\beta_c^2}{4}\,I$, where $I\triangleq\iint e^{-\beta_c(|u|+|v|)-\beta_s|u-v|}\,du\,dv$.

To evaluate $I$, define the inner integral $h(u)\triangleq\int_{\mathbb R}e^{-\beta_c|v|}e^{-\beta_s|u-v|}\,dv$. Then $h$ is even, so $I=2\int_0^\infty e^{-\beta_c u}h(u)\,du$. For $u\ge0$, split the $v$-integral at $0$ and $u$:

\[
h(u)=\underbrace{\frac{e^{-\beta_s u}}{\beta_c+\beta_s}}_{v<0}+\underbrace{\frac{e^{-\beta_c u}-e^{-\beta_s u}}{\beta_s-\beta_c}}_{0\le v\le u}+\underbrace{\frac{e^{-\beta_c u}}{\beta_c+\beta_s}}_{v>u}.
\]

(The middle term is read as $u\,e^{-\beta_c u}$ when $\beta_s=\beta_c$.) Multiply by $e^{-\beta_c u}$ and integrate over $u\ge0$ term by term:

\[
\frac{1}{(\beta_c+\beta_s)^2},\qquad \frac{1}{\beta_s-\beta_c}\Big(\frac{1}{2\beta_c}-\frac{1}{\beta_c+\beta_s}\Big)=\frac{1}{2\beta_c(\beta_c+\beta_s)},\qquad \frac{1}{2\beta_c(\beta_c+\beta_s)} .
\]

Summing and doubling,

\begin{equation}
I=2\left[\frac{1}{(\beta_c+\beta_s)^2}+\frac{1}{\beta_c(\beta_c+\beta_s)}\right]=\frac{2(2\beta_c+\beta_s)}{\beta_c(\beta_c+\beta_s)^2},\qquad \operatorname{Var}(x^\star(0))=\frac{\sigma^2\beta_c(2\beta_c+\beta_s)}{2(\beta_c+\beta_s)^2}.
\end{equation}

\emph{Combine.} With $\operatorname{Var}(\theta(0))=\sigma^2$,

\[
\operatorname{Corr}^2=\frac{\big(\sigma^2\beta_c/(\beta_c+\beta_s)\big)^2}{\sigma^2\cdot\sigma^2\beta_c(2\beta_c+\beta_s)/\big(2(\beta_c+\beta_s)^2\big)}=\frac{2\beta_c}{2\beta_c+\beta_s},\qquad \eta_S(0)=\frac{\beta_s}{2\beta_c+\beta_s}=\frac{\ell_c}{\ell_c+2\ell_s}.
\]
\end{proof}

\paragraph{Monotonicity in the decision-relevance length.}
For the positive-radius regime $vT>2\ell_s$, Theorem~\ref{thm:canonical_optimal_radius}
can be written as
\begin{equation}
r^\star=\frac12\int_{2\ell_s}^{vT}\frac{\ell_c}{\ell_c+s}\,ds.
\label{eq:rstar_integral_representation}
\end{equation}
With $vT$ and $\ell_s$ fixed, the integration limits do not change and
$\ell_c/(\ell_c+s)=1-s/(\ell_c+s)$ is strictly increasing in $\ell_c$
for every $s>0$. Thus $r^\star$ increases strictly with $\ell_c$ in this
regime, proving the fourth comparative static.

\section{Numerical reproducibility details}
\label{app:technical-numerics}
This section gives the settings and estimators for the four figures in
Section~\ref{sec:numerical-evaluation}. The crossover uses the exact coupled
fresh-local rule; the radius curves use scalar variance composition; and
the ring and four-geometry comparison use agentwise Gaussian conditioning with $Q=I$.
Normalized scalar and graph losses use analytical no-information regret
$R_\infty$. The crossover instead uses paired sample means.
Pointwise Monte Carlo means have 95\% normal intervals
$\bar L\pm1.96s_L/\sqrt N$, with sample standard deviation $s_L$.
Conditioning uses linear solves rather than explicit matrix inverses.

\begin{table}[H]
\centering
\caption{Settings for the four numerical illustrations. The shared-resource
model fixes $q=1$; graph covariances have average marginal variance one.}
\label{tab:numerical-parameters}
\small
\begin{tabular}{llll}
\toprule
Figure & principal sweep & fixed setting/size & samples \\
\midrule
\ref{fig:local-global-phase} & 9 $\gamma/q$ values in $[0,10]$ & $n=5,10,25,100$ & $30{,}000$/setting \\
\ref{fig:radius-regret} & $r/\ell_c\in[0,2]$ (1,001 points) & 3 regimes; $D/\ell_c=2$ & $50{,}000$/point \\
\ref{fig:decision-relevance} & $\ell_c=0.7,2,6$ & 64-node ring; $T=6$ & $12{,}000$/check \\
\ref{fig:general-graphs} & all integer radii; $T,\ell_c$ & four 64-node graphs & exact \\
\bottomrule
\end{tabular}
\end{table}

\paragraph{Fresh-local/stale-global crossover (Fig.~\ref{fig:local-global-phase}).}
We fix $q=1$, use $n\in\{5,10,25,100\}$, and sample independent
standard-normal state components ($\sigma^2=1$).
The nine coupling ratios are $\gamma/q\in\{0,0.05,0.1,0.2,0.5,1,2,5,10\}$.
At each setting, 30,000 independent current-state draws give paired no-information and fresh-local
losses, using $u_i=\theta_i/(q+\gamma)$.
Writing their means as $o=\widehat R_\infty$ and
$l=\widehat R_{\rm loc}$, the estimator is
$\widehat\rho^\star=\tfrac12\log[o/(o-l)]$.
Its 95\% delta-method interval is
$\widehat\rho^\star\pm1.96\sqrt{g^\top Cg/N}$, where $C$ is the sample
covariance of the paired losses and
\[
g=\left(-\frac{l}{2o(o-l)},\,\frac{1}{2(o-l)}\right)^\top.
\]
The paired covariance is retained; $o-l\le0$ is rejected rather than clipped.
The zero-coupling boundary is zero and is omitted from the logarithmic plot.

\paragraph{Three radius regimes (Fig.~\ref{fig:radius-regret}).}
The scalar calculation uses $\eta_S(r)=\eta_0e^{-2r/\ell_c}$,
$\tau(r)=r/v$, $\ell_c=1$, and $D=2$, with
$(\eta_0,vT)=(0.25,0.2),(0.7,4),(0.9,100)$.
The radius grid has spacing 0.002. At each point, 50,000 pairs of independent
standard normals are scaled by the square roots of the normalized delay and limited-scope terms in
\eqref{eq:R_decomposition}, then added
and squared. This checks scalar variance composition; it does not simulate
a spatial field.

\paragraph{Decision relevance on a ring (Fig.~\ref{fig:decision-relevance}).}
The graph and covariance are fixed on a 64-node ring, with $T=6$ and latency
0.12 per hop. For $\ell_c\in\{0.7,2,6\}$ the operator is
$K_{ij}=e^{-d_G(i,j)/\ell_c}/\|e^{-d_G(i,\cdot)/\ell_c}\|_2$.
With $Q=I$, exact Gaussian conditioning gives spatial omission; 12,000-sample
checks evaluate selected radii. Each operator has its own analytical
$R_\infty=\tfrac12\operatorname{tr}(K\Sigma_S K^\top)$.
Raw state correlation is computed separately from the fixed covariance.

The covariance in Eq.~\eqref{eq:graph_spectral_covariance} uses
$(\kappa_s,\nu)=(0.75,1.5)$ and is scaled to average marginal variance one.
Every integer radius from 0 through the ring diameter 32 is evaluated.
For each operator, Monte Carlo checks use radii
$\{0,r^\star,\min(r^\star+3,32),32\}$. Independent current-state Gaussian
draws are passed through the exact conditional policies; summed squared
prediction errors are divided by $\operatorname{tr}(K\Sigma_SK^\top)$
to estimate the omission share.

\paragraph{Finite noncanonical graphs (Fig.~\ref{fig:general-graphs}).}
Experiment 6 uses a 64-node path and cycle, an $8\times8$
grid without wraparound, and one connected random geometric graph.
For the latter, NumPy's \texttt{default\_rng(20260839)} draws 64 independent
uniform positions in $[0,1]^2$. Edges join pairs within Euclidean distance
0.22; the construction increases this threshold by 0.02 until connected
(0.22 suffices for this seed). Observation radii use hop distance in every
graph. Their diameters are 63, 32, 14, and 8, respectively.
There is no graph-ensemble average or Monte Carlo estimation in this experiment.

The covariance is proportional to $(0.8^2I+L_G)^{-1.5}$, scaled to average
marginal variance one. The target rows are
$K_{ij}=e^{-d_G(i,j)/\ell_c}/\|e^{-d_G(i,\cdot)/\ell_c}\|_2$ and $Q=I$.
For neighborhood $O_i=\mathcal N_r(i)$, the residual variance is
$K_i\Sigma_S K_i^\top-c_i^\top\Sigma_{S,O_iO_i}^{-1}c_i$,
where $c_i=\Sigma_{S,O_i,:}K_i^\top$; linear solves evaluate this expression.
The sum divided by $\operatorname{tr}(K\Sigma_S K^\top)$ gives $\eta_S(r)$.
Writing $\delta$ for the per-hop delay, so $\tau(r)=\delta r$,
we evaluate $1-e^{-2\delta r/T}(1-\eta_S(r))$ at every integer radius from zero
through the diameter, using reference $(T,\ell_c,\delta)=(5,2,0.12)$.

The plotted sweeps vary $T\in\{1.5,2.5,4,6,9,14\}$ or
$\ell_c\in\{0.7,1.2,2,3.5,6\}$ alone. The archive also retains
per-hop delays $\delta\in\{0.05,0.08,0.12,0.18,0.28\}$ and covariance parameters
$\kappa_s\in\{1.5,1,0.7,0.45,0.3\}$, each varied alone.
Configuration \path{numerics/config/experiment_06.yaml} records experiment
seed 20260828 and graph seed 20260839. The driver is
\path{numerics/experiments/exp06_general_graphs.py}, which calls
\texttt{exp06()} in \path{numerics/experiments/run_all.py}.
From the repository root, run \texttt{python scripts/verify\_graph\_restoration.py}
to rerun that function in scratch space and compare both processed CSVs and
all four raw graph archives under \path{numerics/results/}.
The audit checks positive-definite covariances, conditioning through an
independent residual-error calculation, and exhaustive integer-radius minima.
Its report is \path{docs/graph-restoration-verification.json}.

\paragraph{Reproduction of Figures~\ref{fig:local-global-phase}--\ref{fig:decision-relevance}.}
The \href{https://github.com/ANRGUSC/TINA}{public TINA repository} retains
the original experiment identifiers. Figures~\ref{fig:local-global-phase},
\ref{fig:radius-regret}, and \ref{fig:decision-relevance} correspond to
\path{numerics/config/experiment_02.yaml},
\path{numerics/config/experiment_03.yaml}, and
\path{numerics/config/experiment_05.yaml}, with seeds 20260824, 20260825,
and 20260827, respectively. Their drivers are
\path{numerics/experiments/exp02_local_global_crossover.py},
\path{numerics/experiments/exp03_radius_regret.py}, and
\path{numerics/experiments/exp05_decision_relevance.py}.
Run metadata, processed results, and selected raw arrays are stored under
\path{numerics/results/}; not every random draw is archived.
\path{docs/numerical_method.md} and \path{docs/experiment_parameters.md}
specify conditioning safeguards and environment versions;
\path{numerics/README.md} gives reproduction commands.